\documentclass[12pt]{article}
\usepackage{amssymb,amsmath}
\usepackage{amsthm}
\usepackage[usenames,dvipsnames]{xcolor}
\usepackage[UKenglish]{babel}
\usepackage[UKenglish]{isodate}
\usepackage{cancel}
\usepackage[colorlinks=true]{hyperref}
\usepackage{orcidlink}
\usepackage[bottom=0.5in,top=0.5in,left=0.5in,right=0.5in]{geometry}

\newtheorem{theorem}{Theorem}
\newtheorem{proposition}[theorem]{Proposition}
\newtheorem{lemma}[theorem]{Lemma}
\newtheorem{corollary}[theorem]{Corollary}

\theoremstyle{definition}

\newtheorem{remark}[theorem]{Remark}
\newcommand{\td}{\mathrm{D}_x}
\newcommand{\p}{\partial}

\allowdisplaybreaks[3]

\newcommand*{\pd}
[2]{\mathchoice{\frac{\partial#1}{\partial#2}}
  {\partial#1/\partial#2}{\partial#1/\partial#2}
  {\partial#1/\partial#2}}

\newcommand*{\fd}
[2]{\mathchoice{\frac{\delta#1}{\delta#2}}
  {\delta #1/\delta#2}{\delta#1/\delta#2}{\delta#1/\delta#2}}
\newcommand{\ddx}[1]{\partial_x^{#1}}

\begin{document}

\title{Compatible pairs of Hamiltonian operators\\
  of the first and third orders}

\author{P. Lorenzoni\orcidlink{0000-0001-6171-0821}$^{*}$,
  S. Opanasenko\orcidlink{0000-0002-7956-2473}$^{**}$
  and
  R. Vitolo\orcidlink{0000-0001-6825-6879}$^{***}$}
\maketitle
\vspace{-7mm}
\begin{center}
    ${}^{*}$ Dipartimento di Matematica e Applicazioni, University of Milano-Bicocca\\
  via Roberto Cozzi 53 I-20125 Milano, Italy\\
  INFN sezione di
  Milano-Bicocca\\
  $^{***}$ Dipartimento di Matematica e Fisica ``E. De Giorgi'', Universit\`a
  del Salento\\
    via per Arnesano, 73100 Lecce, Italy\\
  $^{**}$ $^{***}$ INFN Sezione di Lecce\\
$^{**}$ Institute of Mathematics of NAS of Ukra\"\i{}ne\\
  3 Tereshchenkivska Str.\, 01024 Ky\"\i{}v, Ukra\"\i{}ne\\
  e-mails:
\texttt{paolo.lorenzoni@unimib.it}\\
  \texttt{stanislav.opanasenko@le.infn.it}\\
  \texttt{raffaele.vitolo@unisalento.it}
\end{center}

\begin{abstract}
  We compute general compatibility conditions between a weakly nonlocal
  homogeneous Hamiltonian operator and a third-order homogeneous Hamiltonian
  operator. Such operators determine a bi-Hamiltonian structure for many
  integrable PDEs (Korteweg--De Vries, Camassa--Holm, dispersive water waves,
  Dym, WDVV and others). Remarkably, the full set of conditions is purely
  algebraic and the first-order operator is completely determined by commuting
  systems of conservation laws that are Hamiltonian with respect to a
  third-order operator. We illustrate the above results with several examples,
  some of which, concerning WDVV equations, are new.
\end{abstract}

\tableofcontents

\section{Introduction: formulation of the problem and main
  results}

One of the most important ways to define integrable hierarchies of PDEs is
through a compatible pair\footnote{We recall that two Hamiltonian operators
  $A_1$, $A_2$ are compatible if and only if any linear combination
  $\lambda A_1 + \mu A_2$ of them (`pencil') is a Hamiltonian operator.}  of
Hamiltonian operators $A_1$, $A_2$.  Indeed, by Magri's Theorem
\cite{Magri:SMInHEq} an integrable hierarchy of PDEs in $n$~dependent variables
$u^i=u^i(t,x)$, $i=1,\ldots,n$, of the form
\begin{equation*}
  u^i_{t} = A_1^{ij}\left(\fd{H_n}{u^j}\right)
\end{equation*}
is defined by recursively solving the equation
\begin{equation*}
  A_1\left(\fd{H_{n+1}}{u^i}\right) = A_2\left(\fd{H_n}{u^i}\right)
\end{equation*}
for $H_{n+1}$. That implies that all PDEs in the hierarchy are
endowed with an infinite sequence of conserved densities $\{H_n\}$ that
commute with respect to the Poisson brackets defined by both Hamiltonian
operators as
\begin{equation*}
  \{F,G\}_{i} = \int \fd{F}{u^h} A^{hk}_i \fd{G}{u^k}\, \mathrm dx,\qquad i=1,2.
\end{equation*}
That is the essence of the celebrated bi-Hamiltonian formalism for PDEs. In
practice, the applicability of the method, i.e. the well-posedness of the
recursion relations, relies on some additional conditions like the vanishing of
the certain cohomology groups ensuring the existence of solutions at each step
of the iteration.

If we have a closer look at bi-Hamiltonian equations that were found in the
beginning of the research work on integrable PDEs, we notice that many of them
are determined by Hamiltonian operators that are homogeneous with respect to
the grading given by the count of $x$-derivatives in their expressions.  The
study of such operators was initiated by Dubrovin and
Novikov~\cite{DN83,DubrovinNovikov:PBHT}.  The list of the above equations
includes, but is not exhausted by, the KdV equation, the AKNS
equation~\cite{falqui06:_camas_holm,liu05:_defor_hamil}, the Kaup--Broer
system~\cite{Kuper_85}, the Dispersive Water Waves
system~\cite{antonowicz89:_factor_scroed,Kuper_85}, the coupled Dym
system~\cite{antonowicz88:_coupl_harry_dym}, the Monge--Amp\`ere
system\footnote{It is a Monge--Amp\`ere-type equation rewritten as a
  first-order quasilinear system of PDEs,
  see~\cite{mokhov98:_sympl_poiss,OlverNutku1988}}, the WDVV first-order
systems and others (see \cite{LSV:bi_hamil_kdv,vasicek21:_wdvv_hamil,
  OpanasenkoVitolo2024} for a more complete discussion and an extensive list of
references). Basing on remarks in \cite{olver96:_tri_hamil}, we introduced two
types of bi-Hamiltonian equations determined by the following types of
bi-Hamiltonian pairs:
\begin{itemize}
\item bi-Hamiltonian pairs $(A_1,A_2)$ of
  KdV-type~\cite{LSV:bi_hamil_kdv,lorenzoni23:_bih_kdv_monge}:
\begin{equation*}
  A_1 = P_1+R_k, \qquad A_2=Q_1
\end{equation*}
where $P_1$, $Q_1$ are homogeneous first-order Hamiltonian operators,
$R_{k}$ is a homogeneous $k$-st order Hamiltonian operator,
and all operators are mutually compatible. Equations of this type include the
KdV, AKNS, Kaup--Broer, Dispersive Water Waves and Dym equations;
\item bi-Hamiltonian pairs $(A_1,A_2)$ of WDVV-type~\cite{OpanasenkoVitolo2024}:
\begin{equation*}
  A_1 = R_{k}, \qquad A_2=Q_1
\end{equation*}
where $Q_1$ and $R_{k}$ are compatible operators as above. Equations of this
type include the WDVV and Monge--Amp\`ere equations, after rewriting them as
first-order systems of conservation laws.
\end{itemize}

Conditions under which homogeneous operators of the first, second and third
orders are Hamiltonian are well known, see \emph{e.g.} the review by
Mokhov~\cite{mokhov98:_sympl_poiss}.  On the other hand, the compatibility of
two operators is a way more complicated problem.

Compatibility of two homogeneous first-order operators, local and nonlocal, has
been extensively studied (see, for example, the review
paper~\cite{mokhov17:_pencil} and~\cite{mokhov01:_nloc_compat} for the nonlocal
case) as such pairs of operators define the dispersionless limit of dispersive
integrable
hierarchies~\cite{dubrovin01:_normal_pdes_froben_gromov_witten}. Computing the
conditions is not extremely difficult, as they boil down to the Hamiltonian
property for a first-order operator, which depends linearly on one parameter.

The compatibility of homogeneous operators with different degrees of
homogeneity is a relatively simple matter only at low degrees of
homogeneity. For instance, compatible operators of order $1$ and of order $0$
can be found in \cite{mokhov98:_sympl_poiss}; they were used in in
\cite{dellatti:_geomet} in order to study bi-Hamiltonian pairs of
non-homogeneous operators. Compatible operators of the first-order modified by
means of a non-local term defined by isometries were considered
in~\cite{PVV_21_PRSA}.

However, conditions of compatibility of homogeneous operators with different
degrees of homogeneity become much more difficult to compute as the order
increases. In general, the only possibility is the direct calculation of the
Schouten bracket $[P,R_k]$. This is a natural operation on variational
multivectors (see
\cite{Dorfman:DSInNEvEq,IgoninVerbovetskyVitolo:VMBGJS,magri08:_hamil_poiss,
  Olver:ApLGDEq} for an introductory exposition) which involves an increasing
amount of calculations for higher degrees of homogeneity.  For more details on
various possible approaches to compute the Schouten bracket in case of weakly
nonlocal operators,
see~\cite{CLV19,m.20:_weakl_poiss,lorenzoni20:_weakl_schout}.

Some initial attempts to the calculation of compatibility conditions between a
first-order homogeneous operator $P$ and a higher-order operator $R_k$ of the
simplest type can be found
in~\cite{bolsinov:_applic_nijen_iii,konyaev21:_poiss,lorenzoni23:_projec_kdv}
(where $k=3$ and $R_3^{ij}=f^{ij}\partial_x^3$, $(f^{ij})$ is a constant
non-degenerate symmetric matrix) and~\cite{lorenzoni23:_bih_kdv_monge} (where
$k=2$ and $R_2^{ij}=\omega^{ij}\partial_x^2$, $(\omega^{ij})$ is a constant
non-degenerate skew-symmetric matrix).

The main difficulty in the above calculations is to be able to determine a
`minimal' set of conditions that are necessary and sufficient to the vanishing
of the Schouten bracket. Indeed, most of the coefficients are differential
consequences of conditions that are `hidden' inside a few selected coefficients.

In this paper, motivated by the study of bi-Hamiltonian pairs of KdV (see for
instance~\cite{antonowicz89:_factor_scroed,antonowicz88:_coupl_harry_dym,
  Kuper_85}) and WDVV type
(see~\cite{opanasenko25:_wdvv_hamil,FGMN97,PV15,vasicek21:_wdvv_hamil}), we
compute a (possibly minimal) set of compatibility conditions for a first-order
weakly nonlocal homogeneous Hamiltonian operator and a third-order homogeneous
Hamiltonian operator.

It is well known that homogeneous operators are invariant with respect to the
group of diffeomorphisms of the dependent variables.  This suggests one to
reduce one of the two considered operators to a canonical form with respect to
this group. Such forms are known both for local first-order operators (that can
be reduced to constant form thanks to a classical result of Dubrovin and
Novikov~\cite{DN83}) and for general third-order operators. These can be
reduced to the so-called canonical Doyle--Pot\"emin form
\cite{doyle93:_differ_poiss,potemin91:PhDt,potemin97:_poiss}:
\begin{equation*}
R^{ij}_3=\mathrm D_x(f^{ij}\mathrm D_x+c^{ij}_su^s_x)\mathrm D_x.
\end{equation*}
As a consequence of the Hamiltonian property, $f_{ij}$ is a Monge metric of the
form
\begin{equation}\label{eq:13}
  f_{ij} = \phi_{\alpha\beta}\psi^\alpha_i\psi^\beta_j,
\end{equation}
where $\phi$ is a constant non-degenerate symmetric square matrix of order~$n$,
and $\psi _{k}^{\gamma }=\psi _{ks}^{\gamma }u^s+\omega _{k}^{\gamma }$, with
$\psi _{ij}^{\gamma }=-\psi _{ji}^{\gamma }$, is a non-degenerate square
matrix of dimension~$n$, with the constants $\phi_{\alpha\beta}$,
$\psi _{ij}^{\gamma }$ and $\omega _{k}^{\gamma }$ satisfying the relations
\begin{subequations}\label{eq:18}
  \begin{align}
    \phi _{\beta \gamma }(\psi _{il}^{\beta }\psi _{jk}^{\gamma }
    +\psi_{jl}^{\beta }\psi _{ki}^{\gamma }
    +\psi _{kl}^{\beta }\psi _{ij}^{\gamma})=0,
    \\
    \phi _{\beta \gamma }(\omega _{i}^{\beta }\psi _{jk}^{\gamma
    }+\omega _{j}^{\beta }\psi _{ki}^{\gamma }+\omega _{k}^{\beta }\psi
    _{ij}^{\gamma })=0.
  \end{align}
\end{subequations}
The residual group that preserves the canonical forms of $P$ and $R_3$ is the
group of affine transformations of the dependent variables. However, since the
conditions that imply the Hamiltonian property of third-order operators are
known only in the Doyle--Pot\"emin form it is natural to assume $R_3$ to be of
such form.  This choice simplifies a lot the computations, even if one has to
pay the price that the differential geometry of the compatibility conditions is
lost.

On the other hand, it is known that Doyle--Pot\"emin form is preserved by an
extended group, that of the \emph{reciprocal projective
  transformations}~\cite{FPV14,FPV16}, which are projective transformations of
the dependent variables coupled with a $t$-preserving reciprocal transformation
of the independent variables:
\begin{equation*}
  \tilde{u}^i = \frac{T^i_j u^j + T^i_0}{\Delta},\quad \mathrm d\tilde{x} =
  \Delta\,\mathrm dx,\quad \mathrm d\tilde{t}=\mathrm dt,\quad\text{where}\quad
  \Delta = T^0_j u^j + T^0_0\neq0.
\end{equation*}
Since reciprocal projective transformations do not preserve locality we will
consider first-order operators containing a weak nonlocality,
\begin{equation}\label{WNL}
P^{ij}=g^{ij}\mathrm D_x + \Gamma^{ij}_s u^s_x
  + c^{\alpha\beta}w^i_{\alpha k}u^k_x\mathrm D_x^{-1}w^j_{\beta
    l}u^l_x.
\end{equation}
This choice, besides being natural from the point of view of the invariance
group of the problem, is also justified by the existence of several examples
of systems whose bi-Hamiltonian pair contains one of such commuting  pairs
(\emph{e.g.} in \cite{vasicek21:_wdvv_hamil}).

The main results of the paper are listed below.
\begin{enumerate}
\item We calculated the full set of compatibility conditions for the
  operators~$P$ and $R_3$ (this is a nontrivial generalisation of the results
  of~\cite{lorenzoni23:_projec_kdv}), Theorems~\ref{th:comp_nloc} and
  \ref{th:comp_loc}. In particular, we prove that the nonlocal `tail'
  coefficients $w^i_{\alpha j}u^j_x$ of the first-order operator are
  Hamiltonian systems of conservation laws of the third-order operator for
  every $\alpha$:
  \begin{equation}\label{eq:41}
    u^i_t = (w^i_\alpha)_x = w^i_{\alpha k}u^k_x =
    R_3^{ik}\left(\fd{H}{u^k}\right),\qquad i,k=1,\ldots,n.
  \end{equation}
  Such systems are completely described in~\cite{FPV17:_system_cl}
  (see also~\cite{vergallo20:_homog_hamil}): they have the form
  $w^i_\alpha= \psi^i_\gamma Z^\gamma_\alpha$, where $(\psi^i_\gamma)$ is the
  inverse matrix of $(\psi^\gamma_i)$ in~\eqref{eq:13} and
  $Z^\gamma_\alpha = \theta^\gamma_{\alpha k}u^k + \xi^\gamma$ are linear
  functions such that the following linear algebraic equations hold:
  \begin{subequations}\label{eq:16}
    \begin{align}
    &\phi_{\beta\gamma}[\psi_{ij}^{\beta}\theta_{k}^{\gamma j}
      +\psi_{jk}^{\beta}\theta_{i}^{\gamma j}
      +\psi_{ki}^{\beta}\theta_{j}^{\gamma j}]=0,
    \\
    & \phi_{\beta\gamma}[\psi_{ik}^{\beta}\xi^{\gamma j}
      +\omega_{k}^{\beta}\theta_{i}^{\gamma j}
      -\omega_{i}^{\beta}\theta_{k}^{\gamma j}]=0.
    \end{align}
  \end{subequations}
  So, the space of the above Hamiltonian systems is a finite-dimensional vector
  space.
\item We integrated part of the remaining conditions of compatibility,
  and reduced all other
  conditions to algebraic equations, Theorem~\ref{th:compcond_integrated}.  In
  that process, we found the Structure Formula (Corollary~\ref{cor:structure-g})
  \begin{equation*}
    g^{ij} = \psi^i_\gamma Z^{\gamma j} + \psi^j_\gamma Z^{\gamma i}
      -c^{\alpha\beta}w^i_\alpha w^j_\beta
  \end{equation*}
  for the components of the first-order operator (the metric and the nonlocal
  `tail'). The summands $\psi^i_\gamma Z^{\gamma j}$ are again Hamiltonian
  systems of conservation laws of the third-order operator for any
  $j=1,\ldots,n$, with $Z^{\gamma j}=\theta^{\gamma j}_k u^k + \xi^{\gamma j}$,
  where the coefficients $\theta^{\gamma j}_k$, $\xi^{\gamma j}$ fulfill the
  same system of algebraic equations of~\eqref{eq:16}.
\item Finally, in Theorem~\ref{theor:Ham_prop} we proved that the Hamiltonian
  property of the first-order operator implies that the above Hamiltonian
  systems commute: $V^i_{,k}W^k_{,j} = W^i_{,k}V^k_{,j}$, where `$,k$' stands
  for derivative with respect to $u^k$ and $V^i$, $W^i$ are any of the above
  Hamiltonian systems $\psi^i_\gamma Z^{\gamma j}$, $w^i_\alpha = \psi^i_\gamma
  Z^\gamma_\alpha$.
\item We describe the impact of our results on several examples, two of which
  are new: the first-order Hamiltonian operators for first-order WDVV systems
  in dimensions $N=4$ and $N=5$.
\end{enumerate}

We remark that, years ago, we found that a third-order homogeneous Hamiltonian
operator~$R_3$ is completely determined by a quadratic line complex, which is
an algebraic variety in the Pl\"ucker embedding of the projective space with
affine coordinates~$(u^i)$~\cite{FPV14,FPV16}. Moreover, it was
found~\cite{FPV17:_system_cl} that systems of conservation laws that are
Hamiltonian with respect to~$R_3$ can be identified with linear line
congruences, which are also a projective variety in the Pl\"ucker embedding of
a suitable projective space.

After a reciprocal transformation of projective type, all the above algebraic
varieties are preserved. We still do not have a direct proof that first-order
operators~$P$ that are compatible with third-order operators~$R_3$ are
form-invariant with respect to such transformations, but it is very reasonable
to conjecture it (see also the discussion in~\cite{opanasenko25:_wdvv_hamil}).

The paper is organised as follows. In Section~2 we recall the main facts about
first- and third-order homogeneous Hamiltonian operators. These facts are used
in Section~3, that is devoted to the computation of compatibility conditions
and to their consequences. Section~4 contains examples of bi-Hamiltonian
structures of KdV- and WDVV-type. An Appendix is devoted to specific, nontrivial
parts of the computation of compatibility conditions.

\subsection*{Acknowledgements}

PL, SO and RV have been partially supported by the research project
Mathematical Methods in Non Linear Physics (MMNLP) by the Commissione
Scientifica Nazionale -- Gruppo 4 -- Fisica Teorica of the Istituto Nazionale di
Fisica Nucleare (INFN), the Department of Mathematics and Applications of the
Universit\`a Milano-Bicocca, the Department of Mathematics and Physics ``E. De
Giorgi'' of the Universit\`a del Salento, GNFM of the Istituto Nazionale di
Alta Matematica (INdAM). PL and RV have been partially supported by PRIN
2022TEB52W \textit{The charm of integrability: from nonlinear waves to random
  matrices}. RV has been partially supported by ICSC -- Centro Nazionale di
Ricerca in High Performance Computing, Big Data and Quantum Computing, funded
by European Union -- NextGenerationEU and PRIN2020 F3NCPX ``Mathematics for
Industry 4.0''.

\section{Hamiltonian structures and their compatibility}
\label{sec:hamilt-struct}

In this paper, the main problem is to determine the compatibility conditions
between two Hamiltonian operators belonging to two distinct classes: a
first-order weakly nonlinear Hamiltonian operator~$P$ of hydrodynamic type
(see, \emph{e.g.} \cite{ferapontov91:_nloc_hamil}) and a third-order
homogeneous Hamiltonian operator~$R:=R_3$~\cite{DubrovinNovikov:PBHT}. In this
Section we will summarize the main properties of operators in these two
classes.

\subsection{Weakly nonlinear Hamiltonian operator of
  hydrodynamic type}
\label{sec:first-order-hamilt}

Weakly nonlinear Hamiltonian operators of hydrodynamic type
$P=(P^{ij})$:
\begin{equation*}
  P^{ij}=g^{ij}\mathrm D_x + \Gamma^{ij}_s u^s_x
  + c^{\alpha\beta}w^i_{\alpha k}u^k_x\mathrm D_x^{-1}w^j_{\beta
    l}u^l_x
\end{equation*}
are a direct nonlocal generalization of Dubrovin--Novikov homogeneous
Hamiltonian operators of the first order~\cite{DN83}, and were first introduced by
Ferapontov (see the review papers \cite{ferapontov91:_nloc_hamil,F95:_nl_ho}).
Here $g^{ij}$, $\Gamma^{ij}_k$ and $w^i_{\alpha k}$ are functions of the
dependent variables $u^i=u^i(t,x)$, and indices $i$, $j$, $k$, $l$ are in the
range $\{1,\ldots,n\}$. The matrix $(c^{\alpha\beta})$ is square and constant;
in principle, its dimension $m$ is not related with $n$.  The operator $P$
reduces to the Dubrovin--Novikov case when $c^{\alpha\beta}=0$. Throughout this
paper, we assume that the leading coefficient matrix $(g^{ij})$ is
non-degenerate: $\det(g^{ij})\neq 0$. We set $(g_{ij})=(g^{ij})^{-1}$.

We recall that, under the action of a transformation of the dependent variables
$\tilde{u}^i = U^i(u^j)$, $i=1, \dots, n$, $g^{ij}$ transform as the components
of a contravariant $2$-tensor, and the quantities
$\Gamma^i_{jk}=-g_{jp}\Gamma^{pi}_k$ transform as the Christoffel symbols of a
linear connection. Moreover, the quantities $W_{\alpha}=(w^i_{\alpha j})$
transform as the components of a $(1,1)$ tensor, and when contracted with
velocities as in $v^i_\alpha=w^i_{\alpha j}u^j_x$ they transform as the
components of a vector field $v_\alpha=w^i_{\alpha j}u^j_x\pd{}{u^i}$.

It is well-known that the operator~$P$ is Hamiltonian if and only if its associated bracket is Poisson,
\emph{i.e.}, skew-symmetric and Jacobi. It is equivalent to the requirements that~$P$ must be formally
skew-adjoint and the Schouten bracket $[P,P]$ vanishes. It was proved~\cite{F95:_nl_ho}
(see also~\cite{OpanasenkoVitolo2024}) that such
requirements are equivalent to the following conditions:
\begin{enumerate}
\item $g=(g_{ij})$ is a metric, and $\Gamma^i_{jk}$ are the Christoffel symbols
  of the Levi-Civita connection of~$g$; the latter property is equivalent to
  the conditions
  \begin{gather}
    \label{eq:15}
        g^{ij}_{,k} = \Gamma^{ij}_k + \Gamma^{ji}_k,
        \\
        \label{eq:17}
    g^{is}\Gamma^{jk}_s = g^{js}\Gamma^{ik}_s,
  \end{gather}
  where $g^{ij}_{,k}=\pd{g^{ij}}{u^k}$;
\item the $(1,1)$ tensors $W_\alpha=(w^i_{\alpha j})$ are symmetric with
  respect to the metric $g_{ij}$:
  \begin{equation*}
    g^{is}w^j_{\alpha s} = g^{js}w^i_{\alpha s};
  \end{equation*}
\item the tensors $W_\alpha$ fulfill the further property
  \begin{equation}\label{eq:25}
    \nabla_k w^i_{\alpha j} = \nabla_j w^i_{\alpha k};
  \end{equation}
\item the vector fields $v_\alpha=w^i_{\alpha k}u^k_x\pd{}{u^i}$ are a
  commuting family:
  \begin{equation*}
        [v_\alpha,v_\beta] = 0,
  \end{equation*}
  where $[,]$ is the Jacobi bracket~\cite{KruglikovLychagin:GDEq}. When the
  components of $v_\alpha$ and $v_\beta$ are total $x$-derivatives:
  $w^i_{\alpha j}u^j_x = (w^i_\alpha)_x$ (and the same for $v_\beta$), the
  above condition reduces to the commutativity of the matrix product: $W_\alpha
  W_\beta = W_\beta W_\alpha$, or explicitly
  \begin{equation*}
    w^i_{\alpha k}w^k_{\beta j} = w^i_{\beta k}w^k_{\alpha j};
  \end{equation*}
\item the Riemannian curvature\footnote{We follow the Kobayashi--Nomizu
    convention for the Riemannian curvature; see~\cite{OpanasenkoVitolo2024}
    for more details.}  of the Levi-Civita connection of~$g_{ij}$ fulfills
  \begin{equation}\label{eq:34}
    R^{ij}_{kl} = c^{\alpha\beta}\Big( w^i_{\alpha l}w^j_{\beta k}
    - w^i_{\alpha k}w^j_{\beta l}\Big)
      \end{equation}
\item the matrix $(c^{\alpha\beta})$ is symmetric and nondegenerate.
\end{enumerate}
The above conditions have a nice geometric interpretation that is given
in~\cite{ferapontov91:_nloc_hamil,F95:_nl_ho}.

\subsection{Third-order Hamiltonian operator in canonical form}
\label{sec:third-order-hamilt}

The most general form of a third-order homogeneous operator
$R=(R^{ij})$, in accordance with the definition given in~\cite{DubrovinNovikov:PBHT},
is
\begin{equation*}
    R^{ij}=
     f^{ij}\mathrm D_{x}^{3}
      +b_{s}^{ij}u_{x}^{s}\mathrm D_{x}^{2}
  +[c_{s}^{ij}u_{xx}^{s}
  +c_{st}^{ij}u_{x}^{s}u_ {x}^{t}]\mathrm D_{x}+d_{s}^{ij}u_{xxx}^{s}
  +d_{st}^{ij}u_{x}^{s}u_{xx}^{t}
  +d_{srt}^{ij}u_{x}^{s}u_{x}^{r}u_{x}^{t},
\end{equation*}
but in this generality a set of conditions that ensure the
Hamiltonian property of $R$ is unknown.
(As usual we assume that $(f^{ij})$ is a non-degenerate matrix,
and we set $(f_{ij})=(f^{-1})^{ij}$.)
What is known though is, as in the first-order case, $f_{ij}$ transform as
the components of a covariant metric. Moreover, the quantities
 $-\frac{1}{3}f_{jk}d^{ki}_{s}$ transform as the
Christoffel symbols of a linear connection, which is symmetric and flat
as proved independently in~\cite{doyle93:_differ_poiss} and \cite{potemin91:PhDt,
  potemin97:_poiss}. This result has also been generalised to higher-order Hamiltonian operators~\cite{CarletCasati2025}.

On the other hand, there always exists a change of dependent variables
that brings a Hamiltonian operator of the above form
into what we call \emph{Doyle--Pot\"emin canonical form},
\begin{equation*}
R^{ij}=\mathrm D_x(f^{ij}\mathrm D_x+c^{ij}_su^s_x)\mathrm D_x.
\end{equation*}
This form has considerably less coefficients, and the Hamiltonian property is
equivalent to
\begin{subequations}
\begin{gather}
f^{ij}_{,k}=c^{ij}_k+c^{ji}_k,\label{f_der}\\
c^{ij}_kf^{kp}=-c^{pj}_kf^{ki},\label{c_f_sym}\\
c^{ij}_kf^{kp}+c^{jp}_kf^{ki}+c^{pi}_kf^{kj}=0,\label{c_f_cycle}\\
c^{ij}_{k,l}f^{kp} = c^{ip}_kc^{kj}_l - c^{pi}_kc^{kj}_l
- c^{pj}_kc^{ik}_l - c^{pj}_kc^{ki}_l.\label{c_der}
\end{gather}
\end{subequations}
It was proved \cite{FPV14} that, introducing the quantities
$c_{ijk}=f_{iq}f_{jp}c_{k}^{pq}$ the above system simplifies into
\begin{subequations}\label{eq:505}
  \begin{gather}
    c_{ijk}=\frac{1}{3}(f_{ik,j}-f_{ij,k}), \label{eq:506}\\
    f_{ij,k}+f_{jk,i}+f_{ki,j}=0,\label{eq:507}\\
    c_{ijk,l}=-f^{pq}c_{pil}c_{qjk}.\label{eq:733}
  \end{gather}
\end{subequations}
A covariant $2$-tensor $f$ satisfying equations \eqref{eq:507} must be a Monge
metric of a quadratic line complex (see~\cite{FPV14} and reference therein).
This is an algebraic variety that is a section of the Pl\"ucker variety by
means of a quadric. Here, the Pl\"ucker embedding that we are considering is
that of the projective space whose affine chart is $(u^i)$. In particular,
being Monge, $f_{ij}$ is a quadratic polynomial in the field variables~$u$.  It
was further proved
in~\cite{balandin01:_poiss,doyle93:_differ_poiss,potemin91:PhDt,potemin97:_poiss}
that the tensor $f_{ij}$ can be factorized as
\begin{equation*}
  f_{ij} = \phi_{\alpha\beta}\psi^\alpha_i\psi^\beta_j,\quad
  \text{\Big(or, in a matrix form,} \quad f=\Psi\Phi\Psi^\top\Big)
\end{equation*}
where $\phi$ is a constant non-degenerate symmetric matrix of order~$n$, and
$\psi _{k}^{\gamma }=\psi _{ks}^{\gamma }u^s+\omega _{k}^{\gamma }$, with
$\psi _{ij}^{\gamma }=-\psi _{ji}^{\gamma }$, is a non-degenerate square matrix
of dimension~$n$, with the constants $\phi _{\beta \gamma }$,
$\psi _{ij}^{\gamma }$ and $\omega _{k}^{\gamma }$ satisfying the relations~\eqref{eq:18}.

It is known~\cite{FPV17:_system_cl} that a third-order Hamiltonian
operator~$R$ is the Hamiltonian operator of a first-order system of
conservation laws of the form $u^i_t = (V^i)_x$, $i=1,\ldots,n$
(see~\eqref{eq:41}) if
\begin{subequations}
  \label{V}
  \begin{align}\label{e1}
    & f_{im}V^{m}_{,j}=f_{jm}V^m_{,i},\\
    \label{e3}
    & c_{mkl}V^m_{,i}+c_{mik}V^m_{,l}+c_{mli}V^m_{,k}=0,\\
    \label{e2}
    &f_{ks}V^k_{,ij}=c_{smj}V^m_{,i} + c_{smi}V^{m}_{,j},
  \end{align}
\end{subequations}
where $V^i_{,j}=\pd{V^i}{u^j}$.
Nevertheless, it turns out that this result can be improved.

\begin{theorem}\label{th:redundant_condition_3ord_Ham}
  In the system~\eqref{V} the conditions~\eqref{e3} are a consequence
  of~\eqref{e1} and~\eqref{e2}.
\end{theorem}
\begin{proof}
  Let us take the difference of the condition \eqref{e2} with the same
  condition with the indices $s$ and $j$ exchanged:
  \begin{equation}
    \label{eq:181}
    f_{ks}V^k_{,ij} - f_{kj}V^k_{,is} -c_{skj}V^k_{,i} - c_{ski}V^{k}_{,j}
    + c_{jks}V^k_{,i} + c_{jki}V^{k}_{,s} = 0.
  \end{equation}
  We have
  \begin{align*}
    f_{ks}V^k_{,ij} - f_{kj}V^k_{,is}= & (f_{ks}V^k_{,j} - f_{kj}V^k_{,s})_{,i}
    - (f_{ks,i}V^k_{,j} - f_{kj,i}V^k_{,s})
    \\
    =& (c_{ksi} + c_{ski})V^k_{,j} - (c_{kji}+c_{jki})V^k_{,s},
  \end{align*}
  where we used~\eqref{e1}. Then,~\eqref{eq:181} becomes
  \begin{align*}
    &-c_{skj}V^k_{,i} - c_{ski}V^{k}_{,j}
      + c_{jks}V^k_{,i} + c_{jki}V^{k}_{,s} +
      \\
    &\hphantom{ciaociao} (c_{ksi} + c_{ski})V^k_{,j} - (c_{kji}+c_{jki})V^k_{,s} =
    \\
    &-c_{skj}V^k_{,i}
    + c_{jks}V^k_{,i} +
    c_{ksi}V^k_{,j} - c_{kji}V^k_{,s} =
    \\
    &(c_{jsk}+c_{kjs})V^k_{,i}  + c_{jks}V^k_{,i} +
    c_{ksi}V^k_{,j} - c_{kji}V^k_{,s} =
    \\
    &c_{kjs}V^k_{,i} + c_{ksi}V^k_{,j} + c_{kij}V^k_{,s} = 0.
  \end{align*}
\end{proof}

The interested reader can find in \cite{FPV17:_system_cl} the general
expression of the Hamiltonian any system of conservation laws of the above
type.

\section{Conditions of compatibility}
\label{sec:cond-comp}

\subsection{Variational Schouten bracket of operators}
\label{sec:conj-struct-g_1}

In what follows we will assume that $\psi$ denotes a covector-valued density:
$\psi=\psi_i\mathrm du^i\otimes(\mathrm dt\wedge \mathrm dx)$ and
$\phi=\phi^i\pd{}{u^i}$ denotes a (generalized) vector field.  Their
coefficients depend on field variables $(u^i)$ and their $x$-derivatives up to
a finite (although unspecified) order.  Such elements are dual in the geometric
theory of the calculus of variations; conservation laws are uniquely
represented by their characteristic vectors, which are covector-valued
densities, and symmetries are generalized vector fields
\cite{KrasilshchikVinogradov:SCLDEqMP,Olver:ApLGDEq}.

A Hamiltonian operator $A$ can be regarded as a variational bivector,
\emph{i.e.} a bi-differential operator taking two covectors as arguments, with
values into densities: $A(\psi^1,\psi^2)$ is a density, and it is defined up
to total $x$-derivatives. The variational Schouten bracket of two Hamiltonian
operators $[A_1,A_2]$ is a variational three-vector. We will use the following
formula for the variational Schouten bracket (this formula is used by Dubrovin
and Zhang~\cite{dubrovin01:_normal_pdes_froben_gromov_witten}):
\begin{equation*}
  [A_1,A_2](\psi^1,\psi^2)(\psi^3) =
  \left[\ell_{A_1,\psi^1}(A_2(\psi^2))(\psi^3)
    +  \ell_{A_2,\psi^1}(A_1(\psi^2))(\psi^3)  +
    \text{cyclic}(\psi^1,\psi^2,\psi^3)\right].
\end{equation*}
See \cite{KrasilshchikVinogradov:SCLDEqMP,Dorfman:DSInNEvEq,
  IgoninVerbovetskyVitolo:VMBGJS,IgoninVerbovetskyVitolo:FLVDOp,
  magri08:_hamil_poiss,Olver:ApLGDEq} for a comprehensive theory.  The square
brackets in the right-hand side of the above expression mean that the result is
an equivalence class with respect to total $x$-derivatives. In what follows we
will drop the square brackets for ease of reading, although retaining the fact
that computations hold up to total $x$-derivatives.

We recall the formula for the linearization of an operator~$A$ and its adjoint.
If $A(\psi)= a^{ij\sigma}\mathrm D_\sigma\psi_j$ then
\begin{gather*}
  \ell_{A,\psi}(\phi) = \frac{\partial a^{ij\sigma}}{\partial u^k_\tau}
  \mathrm D_\sigma\psi_j\mathrm D_\tau\phi^k,\quad
  \ell^*_{A,\psi^1}(\psi^2) = (-1)^{|\tau|}\mathrm D_\tau
  \left(\psi^2_{i}\frac{\partial a^{ij\sigma}}{\partial u^k_\tau}
  \mathrm D_\sigma\psi^{1}_{j}\right).
\end{gather*}
We recall that if $\Delta$, $\square$ are two operators in total derivatives
then $(\Delta\circ\square)^* = \square^*\circ\Delta^*$. Note that
$(\mathrm D^{-1})^*=-\mathrm D^{-1}$. Moreover, we recall that
\emph{linearization commutes with total derivatives also in case of negative
  exponent}.
There is a linearization formula for the composition
of operators:
\begin{equation*}
\ell_{A\circ B,\psi}(\varphi) = \ell_{A,B\psi}(\varphi) +
A\circ\ell_{B,\psi}(\varphi).
\end{equation*}
In view of next computations we introduce the following notation:
\begin{gather*}
  R = \mathrm D_x\circ F \circ \mathrm D_x,\quad\text{where}\quad
  F = f^{ij}\mathrm D_x + c^{ij}_k u^k_x,
  \\
  P = L + c^{\alpha\beta}N_{\alpha\beta},\quad\text{where}\quad
  L= g^{ij}\mathrm D_x + \Gamma^{ij}_k u^k_{x},\quad
  N_{\alpha\beta}= w^i_{\alpha k}u^k_{x}\mathrm D_x^{-1}w^j_{\beta h}u^h_{x}
\end{gather*}
Then we have
\begin{equation*}
  [P,R] = [L,R] + c^{\alpha\beta}[N_{\alpha\beta},R]
\end{equation*}
where
\begin{align*}
   [L,R](\psi^1,\psi^2,\psi^3) =
  &\ell_{L,\psi^1}(R\psi^2)\psi^3 + \mathrm D_x(\ell_{F,\mathrm D_x\psi^1}(L\psi^2))\psi^3
      + \text{cyclic}(\psi^1,\psi^2,\psi^3)
\end{align*}
and
\begin{align*}
   [N_{\alpha\beta},R](\psi^1,\psi^2,\psi^3) =
  &\ell_{N,\psi^1}(R\psi^2)\psi^3 + \mathrm D_x(\ell_{F,\mathrm D_x\psi^1}(N_{\alpha\beta}\psi^2))\psi^3
      + \text{cyclic}(\psi^1,\psi^2,\psi^3)
\end{align*}

The building blocks of our computation are the linearizations of $L$, $F$
and~$N_{\alpha\beta}$,
\begin{align*}
  &\ell_{L,\psi^1}(\varphi) = g^{ij}_{,k}\mathrm D_x\psi^1_j\varphi^k +
    (\Gamma^{ij}_{h,k}u^h_x\varphi^k + \Gamma^{ij}_{h}\mathrm D_x\varphi^h)\psi^1_j,
  \\
  &\ell_{F,\psi^1}(\varphi) = f^{ij}_{,k}\mathrm D_x\psi^1_j\varphi^k +
  (c^{ij}_{h,k}u^h_x\varphi^k + c^{ij}_{h}\mathrm D_x\varphi^h)\psi^1_j,
  \\
  &\ell_{N_{\alpha\beta},\psi^1}(\varphi) = (w^i_{\alpha k,l}u^k_x\varphi^l + w^i_{\alpha k}\mathrm D_x\varphi^k)
  \mathrm D_x^{-1}(w^j_{\beta h}u^h_x\psi^1_j) +
    w^i_{\alpha h}u^h_x\mathrm D_x^{-1}(w^j_{\beta h,l}u^h_x\varphi^l + w^j_{\beta h}\mathrm D_x\varphi^h)\psi^1_j.
\end{align*}
Introducing the notation
\begin{equation*}
  \tilde{\psi}^i_\alpha = \mathrm D_x^{-1}(w^j_{\alpha k}  u^k_x \psi^i_j),
\end{equation*}
we can rewrite the nonlocal part of the linearization (up to total
$x$-derivatives, or `integrating by parts') as
\begin{equation*}
  \ell_{N_{\alpha\beta},\psi^1}(\varphi) =
  (w^i_{\alpha k,l}u^k_x\varphi^l + w^i_{\alpha k}\mathrm D_x\varphi^k)
  \tilde{\psi}^1_\beta
  - (w^j_{\beta h,l}u^h_x\varphi^l + w^j_{\beta h}\mathrm D_x\varphi^h)
  \tilde{\psi}^1_\alpha
\end{equation*}

In order to compute the components of the three-vector~$[P,R]$,
we will use the algorithm developed in~\cite{CLV19} (see also~\cite{m.20:_weakl_poiss}).
Then, we will gradually simplify the coefficients of the three-vector
using the Hamiltonian properties of~$P$ and~$R$.
It will leave us with a set of conditions
whose vanishing is equivalent to the compatibility of~$P$ and~$R$.

Since the expressions of the coefficients of the three-vector are huge
we limit ourselves to giving a precise description on how to obtain them,
describing the main points and issues of the simplification process
and presenting the conditions under which all coefficients of the three-vector vanish.

\subsection{The algorithm}

The three-vector $[P,R]$ has coefficients that do not determine it uniquely, in
view of the fact that $[P,R]$ is defined up to total $x$-derivatives.  So, our
aim is to obtain a unique, total $x$-derivative free representative of
$[P,R]$. The algorithm in \cite{CLV19} is divided in two parts.

\paragraph{First step: computing the nonlocal part.} We bring all summands
which are linear in $\tilde{\psi}^i_\alpha$'s to the form
\begin{equation}\label{eq:43}
  c^{jp}\tilde{\psi}^a_{\alpha}\mathrm D_x^{n}\psi^b_j\psi^c_p
\end{equation}
where $c^{jp}$ are coefficients depending on~$(u^i,u^i_\sigma)$
and $(a,b,c)\in\{(1,2,3),(3,1,2),(2,3,1)\}$.
We will do this by `integration by parts'; for example,
\begin{equation*}
  c^{jp}\, \tilde{\psi}^1_{\alpha}\psi^2_j\mathrm
  D_x^{n}\psi^3_p
  =
  (-1)^n\mathrm D_x^n\big(c^{jp}\,
  \tilde{\psi}^1_{\alpha}\psi^2_j\big)\mathrm \psi^3_p +
  \text{total $x$-derivative terms},
\end{equation*}
so that any summand which is not of the type of~\eqref{eq:43} is brought to
that type up to total $x$-derivatives.

In the process, we will create new local terms that will be added to the local
part of the three-vector. Note that $[P,R]$ does not contain, in this
particular case, nonlocal terms which are quadratic in
$\tilde{\psi}^i_\alpha$'s.

\paragraph{Second step: computing the local part.} We operate on the local part
of the three-vector as usual. More precisely, we bring the three-vector to a
form that is of zero order with respect to the argument $\psi^3$ using
`integration by parts'; this form is free from total $x$-derivatives. For
example,
\begin{equation*}
  c^{ijp}\, \mathrm D_x^{m_1}\psi^1_i\mathrm D_x^{m_2}\psi^2_j
  \mathrm D_x^{m_3}\psi^3_p
  =
  (-1)^{m_3}\mathrm D_x^{m_3}\big(c^{jp}\,
  \mathrm D_x^{m_1}\psi^1_i\mathrm D_x^{m_2}\psi^2_j\big)\psi^3_p +
  \text{total $x$-derivative terms}.
\end{equation*}

\subsection{Preparatory calculations}
\label{sec:prep-calc}

First of all, we need to compute the four summands in~$[P,R]$.
Then we will apply a cyclic permutation.

\begin{align*}
  \ell_{L,\psi^1} (R(\psi^2))(\psi^3) &=
    g^{pi}_{,k}\ddx{}\psi^1_i
    \big(f^{kj}\ddx{3}\psi^2_j +
    ((f^{kj}_{,h} + c^{kj}_h)u^h_x)\ddx{2}\psi^2_j +
  (c^{kj}_{h,l}u^l_xu^h_x + c^{kj}_hu^h_{2x})\ddx{}\psi^2_j\big)\psi^3_p
    \\
   & +
   \Big(\Gamma^{pi}_{h,k}u^h_x
   \big(f^{kj}\ddx{3}\psi^2_j +
  ((f^{kj}_{,h} + c^{kj}_h)u^h_x)\ddx{2}\psi^2_j +
  (c^{kj}_{h,l}u^l_xu^h_x + c^{kj}_hu^h_{2x})\ddx{}\psi^2_j
  \\
  & + \Gamma^{pi}_{k}   \ddx{}\big(f^{kj}\ddx{3}\psi^2_j +
     ((f^{kj}_{,h} + c^{kj}_h)u^h_x)\ddx{2}\psi^2_j
    +  (c^{kj}_{h,l}u^l_xu^h_x + c^{kj}_hu^h_{2x})\ddx{}\psi^2_j\big)\Big)
    \psi^1_i\psi^3_p,
  \\[1ex]
  \partial_x(\ell_{F,\partial_x(\psi^1)}(L(\psi^2)))(\psi^3) &=
  \ddx{}\Big(f^{pi}_{,k}\ddx{2}\psi^1_i \big(g^{kj}\ddx{}(\psi^2_j)
  + \Gamma^{kj}_h u^h_{x}\psi^2_j\big)
  \\
  & +
    \Big(c^{pi}_{h,k}u^h_x\big(g^{kj}\ddx{}(\psi^2_j)
    + \Gamma^{kj}_h u^h_{x}\psi^2_j\big)
  + c^{pi}_{h}\ddx{}\big(g^{hj}\ddx{}(\psi^2_j)
    + \Gamma^{hj}_l u^l_{x}\psi^2_j\big)\Big)
    \ddx{}(\psi^1_i)\Big)\psi^3_p,
  \\[1ex]
  \ell_{N_{\alpha\beta},\psi^1} (R(\psi^2))(\psi^3) &=
  \Big(  w^p_{\alpha k,l}u^k_x
  \big(f^{lj}\ddx{3}\psi^2_j +
    ((f^{lj}_{,h} + c^{lj}_h)u^h_x)\ddx{2}\psi^2_j
  + (c^{lj}_{h,k}u^k_xu^h_x + c^{lj}_hu^h_{2x})\ddx{}\psi^2_j\big)
   \\
    &   + w^p_{\alpha k}\ddx{}\big(
       f^{kj}\ddx{3}\psi^2_j +
      ((f^{kj}_{,h} + c^{kj}_h)u^h_x)\ddx{2}\psi^2_j
 + (c^{kj}_{h,l}u^l_xu^h_x + c^{kj}_hu^h_{2x})\ddx{}\psi^2_j \big)\Big)
    \tilde{\psi}^1_\beta\psi^3_p
  \\
  & - \Big(w^i_{\beta h,l}u^h_x
    \big(f^{lj}\ddx{3}\psi^2_j +
    ((f^{lj}_{,k} + c^{lj}_k)u^k_x)\ddx{2}\psi^2_j
   + (c^{lj}_{h,k}u^k_xu^h_x + c^{lj}_hu^h_{2x})\ddx{}\psi^2_j\big)
  \\
  & + w^i_{\beta h}\ddx{}\big( f^{hj}\ddx{3}\psi^2_j +
          ((f^{hj}_{,k} + c^{hj}_k)u^k_x)\ddx{2}\psi^2_j
   +
  (c^{hj}_{l,k}u^k_xu^l_x + c^{hj}_ku^k_{2x})\ddx{}\psi^2_j
    \big)\Big)\psi^1_i \tilde{\psi}^3_\alpha,
  \\[1ex]
  \partial_x(\ell_{F,\partial_x(\psi^1)} (N_{\alpha\beta}(\psi^2)))(\psi^3) &=
  \ddx{}\Big(f^{pi}_{,k}\ddx{2}\psi^1_iw^k_{\alpha l}u^l_{x}
    \tilde{\psi}^2_\beta
     +
         \big(c^{pi}_{h,k}u^h_xw^k_{\alpha l}u^l_{x}
         \tilde{\psi}^2_\beta
   + c^{pi}_{h}\ddx{}\big(w^h_{\alpha l}u^l_{x}
     \tilde{\psi}^2_\beta\big)\big)\ddx{}\psi^1_i
     \Big)\psi^3_p.
\end{align*}

Let us introduce the following notation:
\begin{gather*}
  T_{N_{\alpha\beta}} = [N_{\alpha\beta},R](\psi^1,\psi^2,\psi^3) ,
 \qquad
  T_{N_{\alpha\beta}}= \mathfrak N_{\alpha\beta} + \mathfrak L_{\alpha\beta},
\end{gather*}
where  $\mathfrak N_{\alpha\beta}$ contains only nonlocal terms and
$\mathfrak L_{\alpha\beta}$ contains only local terms. We also introduce
\begin{equation*}
  T_{LR} = [L,R](\psi^1,\psi^2,\psi^3).
\end{equation*}
We further split the above expressions into the following summands:
\begin{gather*}
  T_{LR}= T_{LR}^{(123)} + T_{LR}^{(312)} + T_{LR}^{(231)},\qquad
  \\
    \mathfrak N_{\alpha\beta} = \mathfrak N_{\alpha\beta}^{(123)} + \mathfrak N_{\alpha\beta}^{(312)}
  + \mathfrak N_{\alpha\beta}^{(231)},
  \\
    \mathfrak L_{\alpha\beta} = \mathfrak L_{\alpha\beta}^{(123)} + \mathfrak L_{\alpha\beta}^{(312)}
  + \mathfrak L_{\alpha\beta}^{(231)},
\end{gather*}
where the superscript indicate the order of the arguments $\psi^1$, $\psi^2$,
$\psi^3$ in the above four summands with respect to the cyclic permutation.

\subsection{Computing the nonlocal part}
\label{sec:comp-nonl-part}

Applying the first step of the algorithm to
$c^{\alpha\beta}\mathfrak N_{\alpha\beta}$ yields the nonlocal part $T_n$ of the
three-vector $[P,R]$ in canonical form. We will obtain three times the same
components with respect to the three different orderings of the arguments
$\psi^1$, $\psi^2$, $\psi^3$ that we chose in~\eqref{eq:43}, and we write $T_n
= T_n^{(123)}+T_n^{(312)}+T_n^{(231)}$. We list only the
$(3,1,2)$ sector:
\begin{gather*}
  T_{n}^{(312)} = c^{\alpha\beta}(w^q_{\beta k}f^{kj}
  -w^j_{\beta k}f^{kq})\tilde\psi^3_\alpha\mathrm D_x^4\psi^1_j\psi^2_q+
  \\
2c^{\alpha\beta}(-c^{jq}_{k}w^k_{\beta p}+w^q_{\beta
  k}c^{jk}_p+c^{qj}_kw^k_{\beta p}-w^j_{\beta
  k}c^{qk}_p)u^p_x\tilde\psi^3_\alpha\mathrm D_x^3\psi^1_j\psi^2_q+
\\
c^{\alpha\beta}\Big(\mathrm D_x((-w^j_{\beta h}c^{qh}_p+c^{qj}_{k}w^k_{\beta p}
-c^{jq}_kw^k_{\beta p}+w^q_{\beta h}c^{jh}_p)u^p_x)
+S^{jq}_{\beta lp}u^l_xu^p_x\Big)
\tilde\psi^3_\alpha\mathrm D_x^2\psi^1_j\psi^2_q
\\
+c^{\alpha\beta}\Big(\mathrm D_x(S^{jq}_{\beta lp}u^l_xu^p_x)+
c^{\alpha\beta}(-c^{qj}_{k}w^k_{\alpha p}+w^q_{\alpha k}c^{kj}_p-w^j_{\beta
  p,k}f^{kq})u^p_{xxx}\Big)\tilde\psi^3_\alpha\mathrm D_x\psi^1_j\psi^2_q +
\\
c^{\alpha\beta}(w^j_{\beta p,k}-w^j_{\beta k,p})u^p_{xxxx}f^{kq}
\tilde\psi^3_\alpha\psi^1_j\psi^2_q,
\end{gather*}
where
\[
S^{jq}_{\beta lp}:=w^j_{\beta l,k}c^{kq}_p -c^{jq}_{k,l}w^k_{\beta p} +c^{jq}_{l,k}w^k_{\beta p} -w^q_{\beta k,l}c^{kj}_p.
\]

\begin{theorem}\label{th:comp_nloc}
  The nonlocal part $T_n$ of the three-vector $[P,R]$ in canonical form
  vanishes if and only if the following conditions are satisfied:
\begin{subequations}\label{eq13:comp13}
\begin{gather}
  w^j_{\beta p,k} = w^j_{\beta k,p},\label{w_gradient}
  \\
  w^q_{\beta k}f^{kj} = w^j_{\beta k}f^{kq},\label{Commut_f_w}
  \\
  w^j_{\alpha p,k}f^{kq}=w^q_{\alpha k}c^{kj}_p
  - c^{qj}_{k}w^k_{\alpha p}.\label{w_derivative}
\end{gather}
\end{subequations}
\end{theorem}
\begin{proof}
Indeed, $S^{jq}_{\beta lp}+S^{jq}_{\beta pl}$ vanishes identically in view of
the known conditions, see Lemma~\ref{c_der_w_der} in the Appendix.

In turn, the identity
\begin{equation}\label{eq:47}
  -c^{jq}_{k}w^k_{\beta p}+w^q_{\beta k}c^{jk}_p+c^{qj}_kw^k_{\beta p}
  -w^j_{\beta k}c^{qk}_p=0
\end{equation}
is a simple consequence of~\eqref{eq13:comp13}.
\end{proof}
Compatibility implies that we can define new quantities $w^i_\alpha$ such
that $\partial_xw^i_\alpha = w^i_{\alpha j}u^j_x$.
\begin{corollary}
  The potentials $w^i_{\alpha}$ of the coefficients of the nonlocal part of $P$
  can be interpreted, for every index $\alpha$, as fluxes of systems of
  first-order conservation laws that are Hamiltonian with respect to $R$:
  \begin{equation*}
    u^i_t = (w^i_\alpha)_x = w^i_{\alpha ,k}u^k_x
    = R^{ih}\left(\fd{H}{u^h}\right).
  \end{equation*}
  Moreover, we have $w^i_\alpha = \psi^i_\gamma Z^\gamma_\alpha$, where
  $(\psi^i_\gamma)=(\psi^\gamma_i)^{-1}$ and $\psi^\gamma_i$ is a part of the
  decomposition of the Monge metric
  $f_{ij}=\phi_{\alpha\beta}\psi^\alpha_i\psi^\beta_j$ (see
  Subsection~\ref{sec:third-order-hamilt}), and
  $Z^\gamma_\alpha=\eta^\gamma_{\alpha k}u^k + \xi^\gamma_\alpha$ are linear
  functions that solve the linear algebraic system
   \begin{align*}
    &\phi_{\beta\gamma}[\psi_{ij}^{\beta}\eta_{\alpha k}^{\gamma }
      +\psi_{jk}^{\beta}\eta_{\alpha i}^{\gamma }
      +\psi_{ki}^{\beta}\eta_{\alpha j}^{\gamma }]=0,
    \\
    & \phi_{\beta\gamma}[\psi_{ik}^{\beta}\xi^{\gamma }_\alpha
      +\omega_{k}^{\beta}\eta_{\alpha i}^{\gamma }
      -\omega_{i}^{\beta}\eta_{\alpha k}^{\gamma }]=0.
\end{align*}
\end{corollary}
\begin{proof}
  The conditions on the vanishing of the nonlocal part of $[P,R]$ imply that
  the coefficients of the nonlocal part of~$P$ are total derivatives of some
  functions. Then, the other conditions are exactly the conditions found in~\cite{FPV17:_system_cl} for Hamiltonian systems of~$R$ (see also
  Subsection~\ref{sec:third-order-hamilt}).
\end{proof}
\begin{remark}
  Note that we proved in Theorem~\ref{th:redundant_condition_3ord_Ham} that one
  of the conditions in \cite{FPV17:_system_cl} is redundant, namely~\eqref{e3};
  that is the same as~\eqref{eq:47} after lowering the indices.
\end{remark}

\subsection{Computing the local part}
\label{sec:computing-local-part}

We shall compute the canonical form of the local part of the three-vector
$[P,R]$, which is $T_{LR}+ c^{\alpha\beta}\mathfrak L_{\alpha\beta}$. Its expression
can be easily obtained from $T_{LR}^{(123)}+
c^{\alpha\beta}\mathfrak L_{\alpha\beta}^{(123)}$ by means of a cyclic permutation.
A long computation yields the following expressions:
\begin{align*} 
  \begin{split}
    T_{LR}^{(123)} & = g^{pi}_{,k}f^{kj}\ddx{}\psi^1_i\ddx{3}\psi^2_j\psi^3_p
   + g^{pi}_{,k}((f^{kj}_{,h} + c^{kj}_h)u^h_x)
    \ddx{}\psi^1_i\ddx{2}\psi^2_j\psi^3_p
   + g^{pi}_{,k}(c^{kj}_{h,l}u^l_xu^h_x + c^{kj}_hu^h_{2x})
     \ddx{}\psi^1_i\ddx{}\psi^2_j\psi^3_p
  \\ 
    & + \big(\Gamma^{pi}_{h,k}u^h_xf^{kj} + \Gamma^{pi}_{k} \ddx{}(f^{kj})
    + \Gamma^{pi}_{k}((f^{kj}_{,h} + c^{kj}_h)u^h_x)\big)
     \psi^1_i\ddx{3}\psi^2_j\psi^3_p
  \\ 
  & + \big(\Gamma^{pi}_{h,k}u^h_x((f^{kj}_{,m} + c^{kj}_m)u^m_x)
  + \Gamma^{pi}_{k}\ddx{}(((f^{kj}_{,h} + c^{kj}_h)u^h_x))
  + \Gamma^{pi}_{k}(c^{kj}_{h,l}u^l_xu^h_x + c^{kj}_hu^h_{2x})\big)
  \psi^1_i\ddx{2}\psi^2_j\psi^3_p
  \\ 
  & + \Gamma^{pi}_{k} f^{kj}\psi^1_i\ddx{4}\psi^2_j\psi^3_p
  + \big(\Gamma^{pi}_{h,k}u^h_x(c^{kj}_{m,l}u^l_xu^m_x + c^{kj}_mu^m_{2x})
  + \Gamma^{pi}_{k}\ddx{}(c^{kj}_{h,l}u^l_xu^h_x + c^{kj}_hu^h_{2x})\big)
  \psi^1_i\ddx{}\psi^2_j\psi^3_p
    \\ 
  & + \big(\ddx{}(f^{pi}_{,k})g^{kj} + f^{pi}_{,k}\ddx{}(g^{kj})
  + f^{pi}_{,k}\Gamma^{kj}_h u^h_{x}
  + c^{pi}_{h,k}u^h_xg^{kj}
      + c^{pi}_{h}\ddx{}(g^{hj}) + c^{pi}_{h}\Gamma^{hj}_l u^l_{x}\big)
      \ddx{2}\psi^1_i \ddx{}\psi^2_j\psi^3_p
  \\ 
   & + \big(\ddx{}(f^{pi}_{,k})\Gamma^{kj}_h u^h_{x}
     + f^{pi}_{,k} \ddx{}(\Gamma^{kj}_h u^h_{x})
   + c^{pi}_{h,k}u^h_x\Gamma^{kj}_l u^l_{x}
   + c^{pi}_{h}\ddx{}(\Gamma^{hj}_l u^l_{x})\big)
   \ddx{2}\psi^1_i \psi^2_j \psi^3_p
  \\ 
   & + f^{pi}_{,k}g^{kj}\ddx{3}\psi^1_i\ddx{}(\psi^2_j)\psi^3_p
    + f^{pi}_{,k}\Gamma^{kj}_h u^h_{x}\ddx{3}\psi^1_i\psi^2_j\psi^3_p
    + \big( f^{pi}_{,k} g^{kj} + c^{pi}_{h}g^{hj}\big)
      \ddx{2}\psi^1_i\ddx{2}\psi^2_j\psi^3_p
  \\
  & + \big( \ddx{}(c^{pi}_{h,k}u^h_xg^{kj})
   + c^{pi}_{h,k}u^h_x\Gamma^{kj}_l u^l_{x}
     + \ddx{}\big( c^{pi}_{h}\ddx{}(g^{hj})\big)
   + c^{pi}_{h}\ddx{}(\Gamma^{hj}_l u^l_{x})
   +\ddx{}(c^{pi}_{h}\Gamma^{hj}_l u^l_{x})\big)
   \ddx{}(\psi^1_i)\ddx{}\psi^2_j\psi^3_p
  \\ 
   &  + \big( c^{pi}_{h,k}u^h_xg^{kj}+ c^{pi}_{h}\ddx{}(g^{hj})
     + \ddx{}(c^{pi}_{h}g^{hj})
    + c^{pi}_{h}\Gamma^{hj}_l u^l_{x}\big)
     \ddx{}\psi^1_i\ddx{2}\psi^2_j\psi^3_p
  \\ 
  & + \big( \ddx{}(c^{pi}_{h,k}u^h_x\Gamma^{kj}_l u^l_{x})
   + \ddx{}(c^{pi}_{h}\ddx{}(\Gamma^{hj}_l u^l_{x}))\big)
               \ddx{}(\psi^1_i)\psi^2_j\psi^3_p
   + c^{pi}_{h}g^{hj}\ddx{}\psi^1_i\ddx{3}\psi^2_j\psi^3_p
 \end{split}
\end{align*}
and
\begin{equation*}
  \begin{split}
    \mathfrak L_{\alpha\beta}^{(123)}=&\ t_{000}\psi^1_i\psi^2_j\psi^3_p +
    t_{001}\psi^1_i\psi^2_j\td\psi^3_p + t_{002}\psi^1_i\psi^2_j\td^2\psi^3_p
    + t_{003}\psi^1_i\psi^2_j\td^3\psi^3_p + t_{010}\psi^1_i\td\psi^2_j\psi^3_p
    \\
    & + t_{020}\psi^1_i\td^2\psi^2_j\psi^3_p
    + t_{011}\psi^1_i\td\psi^2_j\td\psi^3_p
    + t_{100}\td\psi^1_i\psi^2_j\psi^3_p
    + t_{101}\td\psi^1_i\psi^2_j\td\psi^3_p
    \\
    & + t_{110}\td\psi^1_i\td\psi^2_j\psi^3_p
    + t_{102}\td\psi^1_i\psi^2_j\td^2\psi^3_p
    + t_{200}\td^2\psi^1_i\psi^2_j\psi^3_p
     + t_{201}\td^2\psi^1_i\psi^2_j\td\psi^3_p
    + t_{300}\td^3\psi^1_i\psi^2_j\psi^3_p
  \end{split}
\end{equation*}
where the coefficients are:
\begin{align*}
  t_{000}= &
     - w^i_{\beta h} f^{hj}\ddx{3}\big(w^p_{\alpha k} u^k_x\big)
     - 4\ddx{3}\big(w^i_{\beta h} f^{hj}\big)w^p_{\alpha k} u^k_x
     - 4 \ddx{}\big(w^i_{\beta h} f^{hj}\big)\ddx{2}(w^p_{\alpha k} u^k_x)
     - 6\ddx{2}\big(w^i_{\beta h} f^{hj}\big)
      \ddx{}\big(w^p_{\alpha k} u^k_x\big)
   \\
   &   + \big(w^i_{\beta h,l}u^h_xf^{lj} + w^i_{\beta h}\ddx{}( f^{hj})
        + w^i_{\beta h} ((f^{hj}_{,k} + c^{hj}_k)u^k_x)\big)
        \ddx{2}(w^p_{\alpha m}u^m_x\big)
    \\
    &  + 3\ddx{2}\big(w^i_{\beta h,l}u^h_xf^{lj} + w^i_{\beta h}\ddx{}( f^{hj})
   + w^i_{\beta h} ((f^{hj}_{,k} + c^{hj}_k)u^k_x)\big)w^p_{\alpha m}u^m_x
  \\
  &     + 3\ddx{}\big(w^i_{\beta h,l}u^h_xf^{lj} + w^i_{\beta h}\ddx{}( f^{hj})
     + w^i_{\beta h} (f^{hj}_{,k} + c^{hj}_k)u^k_x\big)
     \ddx{}\big(w^p_{\alpha m}u^m_x\big)
  \\
  &  - \big(w^i_{\beta h,l}u^h_x((f^{lj}_{,k} + c^{lj}_k)u^k_x)
    + w^i_{\beta h}\ddx{}((f^{hj}_{,k} + c^{hj}_k)u^k_x)
    + w^i_{\beta h}(c^{hj}_{l,k}u^k_xu^l_x + c^{hj}_ku^k_{2x})
    \big)\ddx{}\big(w^p_{\alpha m}u^m_x\big)
  \\
  & - 2\ddx{}\big(\big(w^i_{\beta h,l}u^h_x((f^{lj}_{,k} + c^{lj}_k)u^k_x)
    + w^i_{\beta h}\ddx{}((f^{hj}_{,k} + c^{hj}_k)u^k_x)
    + w^i_{\beta h}(c^{hj}_{l,k}u^k_xu^l_x + c^{hj}_ku^k_{2x})
    \big)\big)w^p_{\alpha m}u^m_x
  \\
  & + \big(w^i_{\beta q,l}u^q_x
    (c^{lj}_{h,k}u^k_xu^h_x + c^{lj}_hu^h_{2x})
  + w^i_{\beta h}\ddx{}(c^{hj}_{l,k}u^k_xu^l_x + c^{hj}_ku^k_{2x})
  \big)w^p_{\alpha m}u^m_x
  \\
  & + \big(\ddx{}(f^{pi}_{,k}w^k_{\alpha l}u^l_{x})
    +  c^{pi}_{h,k}u^h_xw^k_{\alpha l}u^l_{x}
  + c^{pi}_{h}\ddx{}(w^h_{\alpha l}u^l_{x}) \big)
      \ddx{}\big(w^j_{\beta m}u^m_x\big)
  \\
  & + 2 \ddx{}
    \big(\ddx{}(f^{pi}_{,k}w^k_{\alpha l}u^l_{x})
    +  c^{pi}_{h,k}u^h_xw^k_{\alpha l}u^l_{x}
  + c^{pi}_{h}\ddx{}(w^h_{\alpha l}u^l_{x}) \big)
  w^j_{\beta m}u^m_x
  \\
  & - f^{pi}_{,h}w^h_{\alpha l}u^l_{x}\ddx{2}\big(w^j_{\beta k}u^k_x\big)
   - 3\ddx{2}\big(f^{pi}_{,h}w^h_{\alpha l}u^l_{x}\big)w^j_{\beta k}u^k_x
  \\
  & - 3\ddx{}\big(f^{pi}_{,h}w^h_{\alpha l}u^l_{x}\big)
      \ddx{}\big(w^j_{\beta k}u^k_x\big)
     - \ddx{}\big(c^{pi}_{h,k}u^h_xw^k_{\alpha l}u^l_{x}
      + c^{pi}_{h}\ddx{}(w^h_{\alpha l}u^l_{x})\big)w^j_{\beta m}u^m_x
  \\
  t_{001}= & - 3 w^i_{\beta h} f^{hj}\ddx{2}\big(w^p_{\alpha k} u^k_x\big)
 - 8 \ddx{}\big(w^i_{\beta h} f^{hj}\big)\ddx{}(w^p_{\alpha k} u^k_x)
  - 6\ddx{2}\big(w^i_{\beta h} f^{hj}\big)w^p_{\alpha k}u^k_x
  \\
  & + 2 \big(w^i_{\beta h,l}u^h_xf^{lj} + w^i_{\beta h}\ddx{}( f^{hj})
    + w^i_{\beta h} ((f^{hj}_{,k} + c^{hj}_k)u^k_x)\big)
    \ddx{}(w^p_{\alpha m}u^m_x\big)
  \\
  &  + 3\ddx{}\big(w^i_{\beta h,l}u^h_xf^{lj} + w^i_{\beta h}\ddx{}( f^{hj})
     + w^i_{\beta h} (f^{hj}_{,k} + c^{hj}_k)u^k_x\big)w^p_{\alpha m}u^m_x
  \\
  &    - \big(w^i_{\beta h,l}u^h_x((f^{lj}_{,k} + c^{lj}_k)u^k_x)
    + w^i_{\beta h}\ddx{}((f^{hj}_{,k} + c^{hj}_k)u^k_x)
    + w^i_{\beta h}(c^{hj}_{l,k}u^k_xu^l_x + c^{hj}_ku^k_{2x})
    \big)w^p_{\alpha m}u^m_x
  \\
  & + 2\big(\ddx{}(f^{pi}_{,k}w^k_{\alpha l}u^l_{x})
    +  c^{pi}_{h,k}u^h_xw^k_{\alpha l}u^l_{x}
    + c^{pi}_{h}\ddx{}(w^h_{\alpha l}u^l_{x}) \big)w^j_{\beta m}u^m_x
  \\
  &
     - 3f^{pi}_{,h}w^h_{\alpha l}u^l_{x}\ddx{}\big(w^j_{\beta k}u^k_x\big)
     - 6 \ddx{}\big(f^{pi}_{,k}w^k_{\alpha l}u^l_{x}\big)w^j_{\beta k}u^k_x
  \\
  t_{002}= & - 3 w^i_{\beta h} f^{hj}\ddx{}\big(w^p_{\alpha k} u^k_x\big)
     - 4 \ddx{}\big(w^i_{\beta h} f^{hj}\big)w^p_{\alpha k} u^k_x
  \\
  & + \big(w^i_{\beta h,l}u^h_xf^{lj} + w^i_{\beta h}\ddx{}( f^{hj})
    + w^i_{\beta h} ((f^{hj}_{,k} + c^{hj}_k)u^k_x)\big)w^p_{\alpha m} u^m_x
    - 3\big(f^{pi}_{,h}w^h_{\alpha l}u^l_{x}\big)w^j_{\beta k}u^k_x
  \\ 
    t_{003} = & - w^i_{\beta h} f^{hj}w^p_{\alpha k} u^k_x
  \\ 
  t_{010} = &
     \big(\ddx{}(f^{pi}_{,k}w^k_{\alpha l}u^l_{x})
    +  c^{pi}_{h,k}u^h_xw^k_{\alpha l}u^l_{x}
  + c^{pi}_{h}\ddx{}(w^h_{\alpha l}u^l_{x}) \big)
      w^j_{\beta m}u^m_x
  \\
  & - 2 f^{pi}_{,h}w^h_{\alpha l}u^l_{x}\ddx{}\big(w^j_{\beta k}u^k_x\big)
    - 3\ddx{}\big(f^{pi}_{,h}w^h_{\alpha l}u^l_{x}\big)w^j_{\beta k}u^k_x
  \\ 
   t_{020}= & - f^{pi}_{,h}w^h_{\alpha l}u^l_{x}w^j_{\beta k}u^k_x
  \\ 
    t_{011} = & - 3f^{pi}_{,h}w^h_{\alpha l}u^l_{x}w^j_{\beta k}u^k_x
  \\
  t_{100} = & - 4 w^i_{\beta h} f^{hj}\ddx{2}(w^p_{\alpha k} u^k_x)
   - 12\ddx{2}\big(w^i_{\beta h} f^{hj}\big)w^p_{\alpha k}u^k_x
   - 12\ddx{}\big(w^i_{\beta h} f^{hj}\big)\ddx{}\big(w^p_{\alpha k}u^k_x\big)
  \\
  & + 3\big(w^i_{\beta h,l}u^h_xf^{lj} + w^i_{\beta h}\ddx{}( f^{hj})
     + w^i_{\beta h} (f^{hj}_{,k} + c^{hj}_k)u^k_x\big)
     \ddx{}\big(w^p_{\alpha m}u^m_x\big)
  \\
  & + 6\ddx{}\big(w^i_{\beta h,l}u^h_xf^{lj} + w^i_{\beta h}\ddx{}( f^{hj})
     + w^i_{\beta h} (f^{hj}_{,k} + c^{hj}_k)u^k_x\big)w^p_{\alpha m}u^m_x
  \\
  & - 2\big(w^i_{\beta h,l}u^h_x((f^{lj}_{,k} + c^{lj}_k)u^k_x)
    + w^i_{\beta h}\ddx{}((f^{hj}_{,k} + c^{hj}_k)u^k_x)
    + w^i_{\beta h}(c^{hj}_{l,k}u^k_xu^l_x + c^{hj}_ku^k_{2x})
    \big) w^p_{\alpha m}u^m_x
    \\
    &  + c^{pi}_{h,k}u^h_xw^k_{\alpha l}u^l_{x}w^j_{\beta m}u^m_{x}
    + c^{pi}_{h}\ddx{}(w^h_{\alpha l}u^l_{x})w^j_{\beta m}u^m_{x}
    + \ddx{}(c^{pi}_{h}w^h_{\alpha l}u^l_{x}w^j_{\beta m}u^m_{x})
  \\ 
  t_{101}=& - 8 w^i_{\beta h} f^{hj}\ddx{}(w^p_{\alpha k} u^k_x)
   - 12\ddx{}\big(w^i_{\beta h} f^{hj}\big)w^p_{\alpha k}u^k_x
  \\
  & + 3\big(w^i_{\beta h,l}u^h_xf^{lj} + w^i_{\beta h}\ddx{}( f^{hj})
     + w^i_{\beta h} (f^{hj}_{,k} + c^{hj}_k)u^k_x\big)w^p_{\alpha m}u^m_x
  \\ 
  t_{102} = & - 4 w^i_{\beta h} f^{hj}w^p_{\alpha k}u^k_x
  \\ 
  t_{110} = & c^{pi}_{h}w^h_{\alpha l}u^l_{x}w^j_{\beta m}u^m_{x}
  \\
  t_{200} = &
   - 6 w^i_{\beta h} f^{hj}\ddx{}\big(w^p_{\alpha k}u^k_x\big)
    - 12\ddx{}\big(w^i_{\beta h} f^{hj}\big)w^p_{\alpha k}u^k_x
  \\
    &+ 3 \big(w^i_{\beta h,l}u^h_xf^{lj} + w^i_{\beta h}\ddx{}( f^{hj})
    + w^i_{\beta h} (f^{hj}_{,k} + c^{hj}_k)u^k_x\big)w^p_{\alpha m}u^m_x
     + \big(f^{pi}_{,k}w^k_{\alpha l}u^l_{x}w^j_{\beta k}u^k_{x}
          + c^{pi}_{h}w^h_{\alpha l}u^l_{x}w^j_{\beta m}u^m_{x}\big)
  \\
  t_{201} = & - 6 w^i_{\beta h} f^{hj}w^p_{\alpha k}u^k_x
  \\ 
  t_{300} = & - 4 w^i_{\beta h} f^{hj}w^p_{\alpha k} u^k_x
\end{align*}

\begin{theorem}\label{th:comp_loc}
  The Hamiltonian operators~$P$ and~$R$ are compatible if and only if
  \begin{itemize}
  \item for every $\alpha$ the coefficients $w^i_{\alpha j}u^j_x$ are total
    $x$-derivatives of functions $w^i_\alpha$, which can be interpreted as
    fluxes of systems of first-order conservation laws
    \begin{equation*}
      u^i_t = (w^i_\alpha)_x = w^i_{\alpha,k}u^k_x
    \end{equation*}
    that are Hamiltonian with respect to
    $R$;
  \item the following conditions hold:
\begin{subequations}\label{eq13:FullCompatibility}
\begin{gather}
  \label{eq:1}
    \Gamma^{ij}_kf^{kp} = \Gamma^{pj}_k f^{ki},
    \\[1ex]
    \label{eq:221}(\Gamma^{ij}_k-\Gamma^{ji}_k)f^{kp}
    + (\Gamma^{pj}_k-\Gamma^{jp}_k)f^{ki}
     + (c^{pi}_kg^{kj} - c^{pj}_k g^{ki}) + (c^{ip}_kg^{kj} - c^{ij}_kg^{kp}) = 0
     \\[1ex]
    \label{eq:8} \big(\Gamma^{ij}_{k,l} + c^{\alpha\beta}w^i_{\alpha k}w^j_{\beta l}\big)f^{kp} = \Gamma^{pj}_kc^{ki}_l - \Gamma^{kj}_l c^{pi}_k,
    \\[1ex]
    \begin{split}\label{eq:118}
    -c^{ij}_{k,l}g^{kp}-c^{pi}_{k,l}g^{kj}-c^{ji}_{k,l}g^{kp}+ c^{pi}_{l,k}g^{kj}+ c^{ij}_{l,k}g^{kp}+ c^{jp}_{l,k}g^{ki}
    \\
    +c^{ki}_l\Gamma^{pj}_k-c^{kp}_l\Gamma^{ij}_k-2c^{ki}_l\Gamma^{jp}_k+2c^{kp}_l\Gamma^{ji}_k -c^{jk}_l\Gamma^{ip}_k+c^{jk}_l\Gamma^{pi}_k
    \\
   -c^{ji}_k\Gamma^{pk}_l+c^{jp}_{k}\Gamma^{ik}_l+c^{\alpha\beta}(w^i_{\alpha k}f^{kj}w^p_{\beta l}-w^j_{\alpha k}f^{kp}w^i_{\beta l})=0,
\end{split}
\end{gather}
\end{subequations}
\end{itemize}
\end{theorem}

\begin{proof}
  We recall that the second step consists in the `integration by parts' with
  respect to the argument~$\psi^3$ in the local three-vector
  $T_{LR}+ c^{\alpha\beta}\mathfrak L_{\alpha\beta}$ in order to obtain a local
  three-vector $T_l$, which is of zero order with respect to that
  argument. That amounts to quite a long computation, which yields the
  expression
\begin{gather*}
  T_l=\sum\limits_{0\leqslant k+l\leqslant 4}a_{kl}\mathrm D_x^k\psi^1_i\mathrm D_x^l\psi^2_j\psi^3_p
\end{gather*}
where most of the coefficients are too cumbersome to display.
The simplest coefficients are $a_{40}$, $a_{22}$ and~$a_{31}$.
The vanishing of~$a_{40}$ (as well as of~$a_{40}$) immediately results in
\begin{equation}
    \label{Commut_Gamma_f}
    \Gamma^{ij}_kf^{kp} = \Gamma^{pj}_k f^{ki},
\end{equation}
vanishing of~$a_{22}$ is equivalent to the condition
\begin{equation}
    \label{Gamma_ji_Gamma_ij}
    3(\Gamma^{ij}_k - \Gamma^{ji}_k)f^{kp} =
    -2(c^{pi}_k g^{kj} - c^{pj}_k g^{ki})
    -c^{ip}_k g^{kj} + c^{jp}_k g^{ki} - c^{ji}_k g^{kp} + c^{ij}_k g^{kp}.
  \end{equation}
  The vanishing of~$a_{31}$ is equivalent to the condition
  \begin{equation}
    \label{eq:26}
       (\Gamma^{ij}_k-\Gamma^{ji}_k)f^{kp}
    + (\Gamma^{pj}_k-\Gamma^{jp}_k)f^{ki}
     + (c^{pi}_kg^{kj} - c^{pj}_k g^{ki}) + (c^{ip}_kg^{kj} - c^{ij}_kg^{kp}) = 0.
  \end{equation}
  Equation~\eqref{Gamma_ji_Gamma_ij} implies equation~\eqref{eq:26}: it is
  enough to replace the differences of Christoffel symbols in~\eqref{eq:26}
  with the help of~\eqref{Gamma_ji_Gamma_ij} to realize that. On the other
  hand, equation~\eqref{eq:26} implies~\eqref{Gamma_ji_Gamma_ij}: subtracting
  from~\eqref{eq:26} the same equation with $i$, $j$ interchanged we
  get~\eqref{Gamma_ji_Gamma_ij}. So, the two conditions are equivalent and each
  of them yields the vanishing of both~$a_{22}$ and~$a_{31}$. Moreover, the
  coefficient $a_{13}$ vanishes in view of these two conditions as well.

  Other coefficients are too cumbersome but we can deduce more from looking at
  their highest terms.  Thus, vanishing of the coefficient of~$u^l_{3x}$
  in~$a_{10}$ results in
\begin{equation}
    \label{Gamma_derivative}
    \big(\Gamma^{ij}_{k,l} +c^{\alpha\beta} w^i_{\alpha k}w^j_{\beta l}\big)f^{kp} = \Gamma^{pj}_kc^{ki}_l - \Gamma^{kj}_l c^{pi}_k,
  \end{equation}
and that of~$u^k_{4x}$ in~$a_{00}$ reduces to
\begin{equation}
    \label{Gamma_kl_Gamma_lk}
    \Gamma^{ij}_{k,l} - \Gamma^{ij}_{l,k} = c^{\alpha\beta}(w^j_{\alpha k}w^i_{\beta l} -  w^i_{\alpha k} w^j_{\beta l}).
  \end{equation}
The latter condition is satisfied in view of~\eqref{Gamma_derivative},
\eqref{Commut_Gamma_f} and~\eqref{c_f_sym}.
Indeed,
\begin{gather*}
  (\Gamma^{ij}_{k,l} + c^{\alpha\beta} w^i_{\alpha k}w^j_{\beta
    l})f^{kp}f_{ps}=(\Gamma^{pj}_kc^{ki}_l-\Gamma^{kj}_lc^{pi}_k)f_{ps}
  \\
  \Gamma^{ij}_{k,l} + c^{\alpha\beta} w^i_{\alpha k}w^j_{\beta
    l}=(\Gamma^{rj}_sc^{si}_l-\Gamma^{sj}_lc^{ri}_s)f_{rk}
  \\
  \Gamma^{ij}_{k,l}+ c^{\alpha\beta} w^i_{\alpha k}w^j_{\beta l}
  - (\Gamma^{ij}_{l,k} + c^{\alpha\beta} w^i_{\alpha l}w^j_{\beta k})=
  (\Gamma^{rj}_sc^{si}_l-\Gamma^{sj}_lc^{ri}_s)f_{rk}-(\Gamma^{rj}_sc^{si}_k-\Gamma^{sj}_kc^{ri}_s)f_{rl}=
  \\
  \Gamma^{rj}_sc^{si}_lf_{rk}-\Gamma^{sj}_lc^{ri}_sf_{rk}+\Gamma^{sj}_lc^{ri}_sf_{kr}-\Gamma^{rj}_sc^{si}_lf_{kr}\equiv
  0.
\end{gather*}
Differentiating~\eqref{Commut_Gamma_f} and killing the derivatives of
Christoffel symbols there with the help of~\eqref{Gamma_derivative} yields a
useful identity
\begin{equation}\label{Gamma_c}
  \Gamma^{ij}_k c^{pk}_l - \Gamma^{pj}_k c^{ik}_l + \Gamma^{kj}_l c^{ip}_k
  - \Gamma^{kj}_l c^{pi}_k = 0,
\end{equation}
which is equivalent to the vanishing of both~$a_{30}$ and~$a_{03}$.  With the
above identities at our disposal we can proceed to simplifying the remaining
coefficients.  Thus, the coefficient~$a_{21}$ readily gives rise to the
condition
\begin{equation}
\begin{split} \label{c_lk_g}
      &-c^{ij}_{k,l}g^{kp}-c^{pi}_{k,l}g^{kj}-c^{ji}_{k,l}g^{kp}
      + c^{pi}_{l,k}g^{kj}+ c^{ij}_{l,k}g^{kp}+ c^{jp}_{l,k}g^{ki}+c^{ki}_l\Gamma^{pj}_k-c^{kp}_l\Gamma^{ij}_k-2c^{ki}_l\Gamma^{jp}_k+2c^{kp}_l\Gamma^{ji}_k\\
      &
      -c^{jk}_l\Gamma^{ip}_k+c^{jk}_l\Gamma^{pi}_k-c^{ji}_k\Gamma^{pk}_l+c^{jp}_{k}\Gamma^{ik}_l
      +c^{\alpha\beta}(w^i_{\alpha k}f^{kj}w^p_{\beta l}-w^j_{\alpha k} f^{kp}w^i_{\beta l})=0,
  \end{split}
\end{equation}
the coefficient~$a_{20}$ is
\begin{align*}
   &   (-\Gamma^{kj}_lc^{pi}_k-\Gamma^{pj}_kc^{ik}_l+c^{ip}_k\Gamma^{kj}_l  + \Gamma^{ij}_{k}c^{pk}_l)u^l_{xx}\\
&+\big(\Gamma^{pj}_{k,m}c^{ki}_l- \Gamma^{ij}_{k,l}c^{kp}_m
       -\Gamma^{kj}_lc^{pi}_{k,m}+\Gamma^{kj}_lc^{pi}_{m,k}
       +c^{\alpha\beta}(w^p_{\alpha k}c^{ki}_lw^j_{\beta m} -w^i_{\alpha k}c^{kp}_lw^j_{\beta m})\big)u^l_xu^m_x,
\end{align*}
where the first term vanishes in view of~\eqref{Gamma_c},
and the second one through symmetrization gives rise to the condition
\begin{equation}\begin{split}\label{Gamma_der_c_der}
\Gamma^{pj}_{k,m}c^{ki}_l- \Gamma^{ij}_{k,l}c^{kp}_m
       -\Gamma^{kj}_lc^{pi}_{k,m}+\Gamma^{kj}_lc^{pi}_{m,k}
+\Gamma^{pj}_{k,l}c^{ki}_m- \Gamma^{ij}_{k,m}c^{kp}_l
       -\Gamma^{kj}_mc^{pi}_{k,l}+\Gamma^{kj}_mc^{pi}_{l,k}\\
+c^{\alpha\beta}(w^p_{\alpha k}c^{ki}_lw^j_{\beta m}-w^i_{\alpha k}c^{kp}_lw^j_{\beta m}
-w^i_{\alpha k}c^{kp}_mw^j_{\beta l}+w^p_{\alpha k}c^{ki}_mw^j_{\beta l})=0,
\end{split}\end{equation}
which vanished identically in view of the known conditions as shown in
Lemma~\ref{lem:RR}.
The coefficient~$a_{11}$ identically vanishes in view of the above conditions,
while the coefficient~$a_{02}$ is equivalent to the
condition~\eqref{Gamma_der_c_der}.  The coefficient~$a_{10}$ is still quite
lengthy after preliminary simplifications, so we opt to present it
part-wise. First of all, the highest coefficient thereof (of~$u^l_{xxx}$)
vanishes identically; the coefficient of~$u^l_{xx}u^m_x$ is not symmetrical,
\begin{gather*}
  -4\Gamma^{ij}_{k,l}c^{pk}_m-5\Gamma^{ij}_{k,l}c^{kp}_m+3\Gamma^{ij}_{k,m}c^{kp}_l
  +4\Gamma^{ij}_{k,m}c^{pk}_l+5\Gamma^{pj}_{k,m}c^{ki}_l-3\Gamma^{pj}_{k,l}c^{ki}_m
  \\
  +4\Gamma^{ij}_kc^{kp}_{l,m}+4\Gamma^{ij}_kc^{pk}_{l,m}
  -4\Gamma^{ij}_kc^{pk}_{m,l}-4\Gamma^{ij}_kc^{kp}_{m,l}
  -5\Gamma^{kj}_lc^{pi}_{k,m}
  \\
  +3\Gamma^{kj}_mc^{pi}_{k,l}+\Gamma^{kj}_mc^{pi}_{l,k}+\Gamma^{kj}_lc^{pi}_{m,k}
  +4\Gamma^{pj}_kc^{ki}_{l,m}-4\Gamma^{pj}_kc^{ki}_{m,l}
  \\
  +c^{\alpha\beta}\Big(5c^{ki}_lw^p_{\alpha k}w^j_{\beta
    m}-3c^{ki}_mw^p_{\alpha k}w^j_{\beta l}+3c^{kp}_lw^i_{\alpha k}w^j_{\beta
    m} -5c^{kp}_mw^i_{\alpha k}w^j_{\beta l}
  \\
  -4c^{pk}_mw^i_{\beta k}w^j_{\alpha l}-4c^{ip}_kw^k_{\alpha l}w^j_{\beta m}
  +4c^{pi}_kw^j_{\alpha m}w^k_{\beta l}+4c^{ik}_lw^p_{\alpha k}w^j_{\beta
    m}\Big),
\end{gather*}
and therefore it better to write down its symmetrisation and anti-symmetrisation with respect to~$l,m$.
The symmetrical part gives rise to the condition~\eqref{Gamma_der_c_der}, and its antisymmetrical counterpart to
\begin{gather*}
-\Gamma^{ij}_{k,l}c^{pk}_m+\Gamma^{ij}_{k,m}c^{kp}_l
+\Gamma^{ij}_{k,m}c^{pk}_l-\Gamma^{pj}_{k,l}c^{ki}_m
-\Gamma^{ij}_{k,l}c^{kp}_m
+\Gamma^{pj}_{k,m}c^{ki}_l\\
+\Gamma^{ij}_kc^{kp}_{l,m}+\Gamma^{ij}_kc^{pk}_{l,m}-\Gamma^{ij}_kc^{pk}_{m,l}-\Gamma^{ij}_kc^{kp}_{m,l}-\Gamma^{kj}_lc^{pi}_{k,m}
+\Gamma^{pj}_kc^{ki}_{l,m}-\Gamma^{pj}_kc^{ki}_{m,l}+\Gamma^{kj}_mc^{pi}_{k,l}\\
+c^{\alpha\beta}\left(c^{kp}_lw^i_{\alpha k}w^j_{\beta m}-c^{ki}_mw^p_{\alpha k}w^j_{\beta l}+c^{pk}_lw^i_{\alpha k}w^j_{\beta m}
-c^{pk}_mw^i_{\alpha k}w^j_{\beta l}+c^{ki}_lw^p_{\alpha k}w^j_{\beta m}-c^{kp}_mw^i_{\alpha k}w^j_{\beta l}\right)=0,
\end{gather*}
which vanishes in view of Lemma~\ref{lemma:re}.
The coefficient of~$u^l_xu^m_xu^n_x$ in~$a_{10}$ is
\begin{gather*}
\frac{\p}{\p u^n}\left(-\Gamma^{ij}_{k,l}c^{kp}_m
+\Gamma^{pj}_{k,m}c^{ki}_l
-\Gamma^{kj}_lc^{pi}_{k,m}+\Gamma^{kj}_mc^{pi}_{l,k}
+c^{\alpha\beta}(c^{ki}_mw^p_{\alpha k}w^j_{\beta l}-c^{kp}_lw^i_{\alpha k}w^j_{\beta m})\right)
+c^{\alpha\beta}(S^{ip}_{\alpha nl}w^j_{\beta m}+S^{ij}_{\alpha nl}w^p_{\beta m}).
\end{gather*}
The first part thereof vanishes in view of~\eqref{Gamma_der_c_der}, while the
latter vanishes due to Lemma~\ref{c_der_w_der}. Finally, the most lengthy
coefficient~$a_{00}$ simplifies marvelously to just
  \begin{align*}
    & c^{\alpha\beta}\ddx{}(S^{ip}_{\alpha nl}w^j_{\beta m}u^l_xu^m_xu^n_x),
  \end{align*}
which is identically zero due to Lemma~\ref{c_der_w_der} as well. Finally,
we can show that the coefficient~$a_{01}$ simplifies to~\eqref{Gamma_kl_Gamma_lk},
and thus vanishes identically as well.
\end{proof}

An interesting Corollary provides a new structure arising from the first-order
operator~$P$.
\begin{corollary}\label{cor:assoc_alg}
  The cotangent space of the space of dependent variables $(u^i)$ is endowed by
  the Christoffel symbols $\Gamma^{ij}_k$ with a structure of associative
  algebra (without unity):
  \begin{equation}\label{eq:23}
    \Gamma^{ij}_s\Gamma^{sk}_l - \Gamma^{ik}_s \Gamma^{sj}_l = 0.
  \end{equation}
\end{corollary}
\begin{proof}
  Indeed, the equation~\eqref{Gamma_kl_Gamma_lk} implies that the
  equation~\eqref{eq:34} is equivalently rewritten as~\eqref{eq:23}.
\end{proof}

\subsection{The structure of \texorpdfstring{$g$}{g}}


We prove a structure formula for $g$ with a series of lemmas. Our initial aim
is to integrate the conditions~\eqref{Gamma_derivative}
and~\eqref{Gamma_kl_Gamma_lk} out of all compatibility conditions of~$P$
and~$R$.

\begin{lemma}
There exists a matrix of functions $r^{ij}$ such that
  \begin{equation}\label{eq:10}
    \Gamma^{ij}_k + c^{\alpha\beta}w^i_{\alpha k} w^j_\beta = r^{ij}_{,k}.
  \end{equation}
\end{lemma}
\begin{proof}
  Indeed,
  \begin{equation*}
    (\Gamma^{ij}_k + c^{\alpha\beta}w^i_{\alpha k} w^j_\beta)_{,l} -
    (\Gamma^{ij}_l + c^{\alpha\beta}w^i_{\alpha l} w^j_\beta)_{,k} = 0,
  \end{equation*}
  in view of the condition of compatibility~\eqref{Gamma_kl_Gamma_lk}.
\end{proof}
Hence, by introducing a matrix of functions $r^{ij}$ as above we integrate the
condition~\eqref{Gamma_kl_Gamma_lk}
\begin{lemma}
  We have the identity
  \begin{equation}
    \label{eq:19}
    g^{ij} + c^{\alpha\beta}w^i_\alpha w^j_\beta = r^{ij} + r^{ji}.
  \end{equation}
\end{lemma}
\begin{proof}
  It is enough to symmetrize the equality~\eqref{eq:10} with respect to the
  indices $i$, $j$.
\end{proof}
The above Lemma shows that the formula~\eqref{eq:19} integrates the
equation~\eqref{eq:15}.

\begin{lemma}
  The system of equations~\eqref{eq:1} and~\eqref{eq:8} is equivalent to the
  system:
  \begin{subequations}\label{eq:20}
  \begin{align}
    \label{eq:125}
    &f^{lk}r^{ij}_{,k} = f^{ik}r^{lj}_{,k},
    \\ \label{eq:185}
    & f^{kp}r^{ij}_{,kl} = r^{pj}_{,k}c^{ki}_l - r^{kj}_{,l}c^{pi}_k.
  \end{align}
  \end{subequations}
\end{lemma}
\begin{proof}
  Indeed, the first equation follows from the identities
  \begin{displaymath}
    f^{lk}\Gamma^{ij}_{k} = f^{ik}\Gamma^{lj}_{k}, \quad
    f^{lk}w^{i}_{\alpha k} = f^{ik}w^{l}_{\alpha k}.
  \end{displaymath}
  The second equation is obtained by rewriting the condition~\eqref{Gamma_derivative}
  using $r^{ij}$ as a variable:
  \begin{equation*}
      (r^{ij}_{,kl} - c^{\alpha\beta}w^i_{\alpha k,l}w^j_{\beta})f^{kp}
    =  (r^{pj}_{,k} - c^{\alpha\beta}w^p_{\alpha k} w^j_\beta)c^{ki}_l
        - (r^{kj}_{,l} - c^{\alpha\beta}w^k_{\alpha l} w^j_\beta)c^{pi}_k.
  \end{equation*}
  Since
  \begin{equation*}
    c^{\alpha\beta}(w^i_{\alpha k,l}f^{kp} + w^k_{\alpha l}c^{pi}_k
    - w^p_{\alpha k}c^{ki}_l)w^j_\beta = 0,
\end{equation*}
and the arguments can be reversed, we have the result.
\end{proof}

Note that the above potential function $r^{ij}$ is not a gradient with respect
to any of the involved pseudo-Riemannian metrics $f_{ij}$, $g_{ij}$. However,
there is the following interesting relation.
\begin{proposition}
  We have the identity
  \begin{equation*}
    (f_{bk}r^{kj})_{,a} - (f_{ak}r^{kj})_{,b} = (f_{kb,a} - f_{ka,b})r^{kj}
    = 3r^{kj}c_{kab}.
  \end{equation*}
\end{proposition}
\begin{proof}
  Indeed, the lower indices version of~\eqref{eq:125} is $f_{bi}r^{ij}_{,a}
  = f_{ai}r^{ij}_{,b}$, from which the result follows.
\end{proof}

Let us introduce the parametrization $r^{ij}=\psi^i_\gamma Z^{\gamma j}$.
The equation~\eqref{eq:185} transforms into the equation
  \begin{equation*}
    Z^{\gamma j}_{,hk}=0,
  \end{equation*}
and thus $Z^{\gamma r} = \eta^{\gamma r}_k u^k + \xi^{\gamma r}$ for some constants~$\eta^{\gamma r}_k$ and~$\xi^{\gamma r}$.
\begin{proposition}\label{theor:integration}
 The equation~\eqref{eq:125} is equivalent to the two linear algebraic equations
  \begin{gather*}
    \phi_{\beta\gamma}[\psi_{ij}^{\beta}\eta_{k}^{\gamma r}
      +\psi_{jk}^{\beta}\eta_{i}^{\gamma r}
      +\psi_{ki}^{\beta}\eta_{j}^{\gamma r}]=0,
    \\
     \phi_{\beta\gamma}[\psi_{ik}^{\beta}\xi^{\gamma r}
      +\omega_{k}^{\beta}\eta_{i}^{\gamma r}
      -\omega_{i}^{\beta}\eta_{k}^{\gamma r}]=0.
  \end{gather*}
\end{proposition}
\begin{proof}
In \cite{FPV17:_system_cl} a closed form solution of the system
\begin{align}
  \label{eq:21}
  &f^{lk}w^{i}_{k} = f^{ik}w^{l}_{k}
  \\\label{eq:22}
  & w^i_{kl}f^{kp} + w^k_{l}c^{pi}_k
    - w^p_{k}c^{ki}_l = 0
\end{align}
is presented by means of the parametrization $w^i = \psi^i_\gamma Z^\gamma$. In
particular, it is proved that the last equation~\eqref{eq:22} is equivalent to
$Z^\gamma_{,ij}=0$, thus proving that
$Z^\gamma = \eta^{\gamma}_k u^k + \xi^{\gamma}$ is a vector of linear
functions. The equation \eqref{eq:21} is proved to be equivalent to the
(linear) algebraic system
 \begin{align*}
    &\phi_{\beta\gamma}[\psi_{ij}^{\beta}\eta_{k}^{\gamma }
      +\psi_{jk}^{\beta}\eta_{i}^{\gamma }
      +\psi_{ki}^{\beta}\eta_{j}^{\gamma }]=0,
    \\
    & \phi_{\beta\gamma}[\psi_{ik}^{\beta}\xi^{\gamma }
      +\omega_{k}^{\beta}\eta_{i}^{\gamma }
      -\omega_{i}^{\beta}\eta_{k}^{\gamma }]=0.
\end{align*}

Here, repeating the above argument using the parametrization $r^{ij} =
\psi^i_\gamma Z^{\gamma j}$ we obtain the result.
\end{proof}

In other words, the matrix $(r^{ij})$ can be interpreted as a family of $n$
systems of first-order conservation laws, indexed by $j$, which is Hamiltonian
with respect to the operator~$R$:
\begin{equation*}
  u^i_t = (r^{ij})_x = R^{ik}\left(\fd{H^{j}}{u^k}\right),
\end{equation*}
where $H^j$ is the Hamiltonian density (see~\cite{FPV17:_system_cl} for more
details); the fluxes have the form $r^{ij}=\psi^i_\alpha
Z^{\alpha j}$ described in the proof of the above Theorem.

The following Structure Formula for the metric of the first-order operator $P$
is a straightforward consequence of~\eqref{eq:19}.
\begin{corollary}[Structure Formula]\label{cor:structure-g}
  We have
  \begin{equation*}
    g^{ij} = \psi^i_\gamma Z^{\gamma j} + \psi^j_\gamma Z^{\gamma i}
    -c^{\alpha\beta}w^i_\alpha w^j_\beta,
  \end{equation*}
  where $w^i_\alpha$ are potentials of the coefficients of the nonlocal part in
  the expression~\eqref{WNL}: $w^i_{\alpha k}=w^i_{\alpha , k}$.
\end{corollary}

The above results allow us to prove a conjecture that was made
in~\cite{OpanasenkoVitolo2024}.
\begin{corollary}
  The metric $g^{ij}$ is of the form
  \begin{equation*}
    g^{ij} = \psi^i_\alpha Q^{\alpha\beta}\psi^j_\beta,
  \end{equation*}
  where $Q^{\alpha\beta}$ is the matrix of quadratic functions defined by
  \begin{equation*}
    Q^{\alpha\beta} = Z^{\alpha j}\psi^\beta_j + \psi^\alpha_i Z^{\beta i} -
    c^{\gamma\delta} Z^\alpha_\gamma Z^\beta_\delta,
  \end{equation*}
  where $w^i_\alpha = \psi^i_\gamma Z^\gamma_\alpha$.
\end{corollary}

Moreover, taking into account the results of this Section the main result can
be formulated as follows.

\begin{theorem}\label{th:compcond_integrated}
  Given the two Hamiltonian operators $P$, $R$ as in
  Subsections~\ref{sec:first-order-hamilt} and~\ref{sec:third-order-hamilt},
  the compatibility conditions $[P,R]=0$ are equivalent to:
  \begin{enumerate}
  \item For every $\alpha$ the vector function $(w^i_\alpha)$ is the vector of
    fluxes of a system of conservation laws which are Hamiltonian with respect
    to $R$.
  \item   \begin{equation*}
    g^{ij} = \psi^i_\gamma Z^{\gamma j} + \psi^j_\gamma Z^{\gamma i}
    -c^{\alpha\beta}w^i_\alpha w^j_\beta,
  \end{equation*}
  where $Z^{\gamma j} = \eta^{\gamma j}_k u^k + \xi^{\gamma j}$ are matrices of
  linear functions whose coefficients fulfill the (linear) algebraic system
  \begin{align*}
    &\phi_{\beta\gamma}[\psi_{ij}^{\beta}\eta_{k}^{\gamma j}
      +\psi_{jk}^{\beta}\eta_{i}^{\gamma j}
      +\psi_{ki}^{\beta}\eta_{j}^{\gamma j}]=0,
    \\
    & \phi_{\beta\gamma}[\psi_{ik}^{\beta}\xi^{\gamma j}
      +\omega_{k}^{\beta}\eta_{i}^{\gamma j}
      -\omega_{i}^{\beta}\eta_{k}^{\gamma j}]=0.
\end{align*}
\item
  \begin{displaymath}
(\Gamma^{ij}_k-\Gamma^{ji}_k)f^{kp}
    + (\Gamma^{pj}_k-\Gamma^{jp}_k)f^{ki}
     + (c^{pi}_kg^{kj} - c^{pj}_k g^{ki}) + (c^{ip}_kg^{kj} - c^{ij}_kg^{kp}) = 0
   \end{displaymath}
 \item
       \begin{align*}\label{eq:11}
    &-c^{ij}_{k,l}g^{kp}-c^{pi}_{k,l}g^{kj}-c^{ji}_{k,l}g^{kp}+ c^{pi}_{l,k}g^{kj}+ c^{ij}_{l,k}g^{kp}+ c^{jp}_{l,k}g^{ki}
    \\
    &+c^{ki}_l\Gamma^{pj}_k-c^{kp}_l\Gamma^{ij}_k-2c^{ki}_l\Gamma^{jp}_k+2c^{kp}_l\Gamma^{ji}_k -c^{jk}_l\Gamma^{ip}_k+c^{jk}_l\Gamma^{pi}_k
    \\
   &-c^{ji}_k\Gamma^{pk}_l+c^{jp}_{k}\Gamma^{ik}_l+c^{\alpha\beta}(w^i_{\alpha k}f^{kj}w^p_{\beta l}-w^j_{\alpha k}f^{kp}w^i_{\beta l})=0.
\end{align*}
\end{enumerate}
\end{theorem}

After having simplified the compatibility equations, we can now ask ourselves
if it is possible to provide an interpretaion of the Hamiltonian properties
of~$P$ written down in Subsection~\ref{sec:first-order-hamilt} in terms of the
above Hamiltonian systems of $R$.  We have the
following results.

\begin{proposition}
  The equation~\eqref{eq:25} is equivalent to the requirement that the vector
  of fluxes $\{w^i_\alpha\mid i=1,\ldots, n\}$ and
  $\{r^{hk}\mid h=1,\ldots,n\}$ commute for every choice of $\alpha$ and $k$.
\end{proposition}
\begin{proof}
  The equation~\eqref{eq:25} is equivalent to
  \begin{equation*}
    \Gamma^{kj}_p w^q_{\alpha k} = \Gamma^{qj}_k w^k_{\alpha p}.
  \end{equation*}
  Then, in view of formula~\eqref{eq:10} and the commutativity of $w_\alpha^i$
  with every other $w_\beta^i$, we obtain the result.
\end{proof}

\begin{proposition}
  The equation~\eqref{eq:34} is equivalent to the requirement that
  the vector of fluxes $\{r^{hk}\mid h=1,\ldots,n\}$ and $\{r^{hj}\mid
  h=1,\ldots,n\}$ commute for every choice of $k$ and $j$.
\end{proposition}
\begin{proof}
  As we have seen (Corollary~\ref{cor:assoc_alg}), the
  equation~\eqref{Gamma_kl_Gamma_lk} implies the equation~\eqref{eq:23}. By the
  formula~\eqref{eq:10}, the previous Proposition and the commutativity of the
  fluxes $w_\alpha^i$ we obtain the result.
\end{proof}

We are now in position to reformulate the Hamiltonian property of~$P$ in terms
of mutually commuting systems of conservation laws of the third-order
compatible Hamiltonian operator $R$.

\begin{theorem}[Hamiltonian property]\label{theor:Ham_prop}
  Let $P$, $R$ be a compatible pair of Hamiltonian operators described in
  Subsections~\ref{sec:first-order-hamilt} and~\ref{sec:third-order-hamilt},
  respectively. Then, the Hamiltonian property of $P$ is equivalent to the
  following requirements:
  \begin{enumerate}
  \item $g^{ik}w_{\alpha k}^j = g^{jk}w_{\alpha k}^i$;
  \item $g^{ik}r_{,k}^{jl} = g^{jk}r_{,k}^{il}$;
  \item The two families of vector of fluxes $\{w^i_\alpha\mid i=1,\ldots, n\}$
    and $\{r^{hk}\mid h=1,\ldots,n\}$ are commuting, both between elements of
    the same family and between elements of the two distinct families:
    \begin{equation*}
      r^{ik}_{,p} r^{ph}_{,j} = r^{ih}_{,p} r^{pk}_{,j},\quad
      r^{ik}_{,p} w^{p}_{\alpha,j} = w^{i}_{\alpha,p} r^{pk}_{,j},\quad
      w^{i}_{\alpha,p} w^{p}_{\beta,j} = w^{i}_{\beta,p} w^{p}_{\alpha,j}.
    \end{equation*}
  \end{enumerate}
\end{theorem}

\section{Examples}
\label{sec:examples}

In this section, we consider various examples of compatible pairs $P$, $R$ and
determine the structure of $P$, according to the Structure Formula
(Corollary~\ref{cor:structure-g}) by providing $Z^{\gamma i}$,
$c^{\alpha\beta}$ and $w^i_\alpha$ with respect to a chosen decomposition
$f_{ij}=\phi_{\alpha\beta}\psi^\alpha_i\psi^\beta_j$.

Some of the examples can be found in the literature. In the past, compatibility
was checked by pen and paper (see, \emph{e.g.}
\cite{antonowicz88:_coupl_harry_dym,antonowicz89:_factor_scroed,FGMN97,kalayci97:_bi_hamil_wdvv,kalayci98:_alter_hamil_wdvv});
recently, new theoretical advances~\cite{CLV19} and computer algebra
packages~\cite{m.20:_weakl_poiss} allowed large compatibility calculations,
also when $P$ is nonlocal
\cite{PV15,vasicek21:_wdvv_hamil,vasicek22:_wdvv,OpanasenkoVitolo2024}.

However, the results of this paper enable us to find the `fine structure' of
the metric of $P$, and to check the compatibility by means of the simplified,
purely algebraic equations to which we reduced the problem. Indeed,
we could not manage to compute compatibility for the new WDVV examples below
using the above computer algebra packages for Schouten bracket. On the other
hand, the compatibility conditions that we found in this paper are easy to
program and enable us to compute in much higher dimensions.

We will provide two categories of examples.

The dispersive examples come from the classification in
\cite{LSV:bi_hamil_kdv}, and are of KdV-type. Here we just confirm the
Structure formula by providing the expression of $Z^{\alpha i}$ for some
examples in the classification.

The dispersionless examples are endowed
by bi-Hamiltonian structures of WDVV-type \cite{OpanasenkoVitolo2024}. Three
examples come from that paper, again we provide the expression of $Z^{\gamma
  i}$. The expression of $Z^\gamma_\alpha$ was already known; what was not
known is the fact that such systems enter in the Structure Formula and
contribute to the definition of the metric of the first-order Hamiltonian
operator.

Two further dispersionless examples are completely new. They arise in WDVV
equations in dimension $N=4$ and $N=5$, after rewriting them as commuting
quasilinear systems of first-order conservation laws. It was recently proved in
\cite{opanasenko25:_wdvv_hamil} that such systems are endowed with a
third-order Hamiltonian structure, and it was conjectured that they are endowed
also with a first-order Hamiltonian structure. Our results below confirm the
conjecture.

In all examples, we checked that commutativity mentioned
in Theorem~\ref{theor:Ham_prop} holds.

\subsection{Bi-Hamiltonian structures of KdV-type}

\paragraph{Dispersive water waves.} The system
\begin{gather*}
u^1_t=\frac14u^2_{xxx}+\frac12u^2u^1_x+u^1u^2_x,\\
u^2_t=u^1_x+\frac32u^2u^2_x
\end{gather*}
is equivalent to dispersive water waves system up to a Miura transformation
\cite{antonowicz89:_factor_scroed}.
Among others, it is bi-Hamiltonian with respect to operators~$R+P_1$ and~$P_2$,
where $R$ is a third-order homogeneous operator defined by the Monge metric with
\begin{gather*}
\psi=\begin{pmatrix} 1 & 0 \\ 0 & 1 \end{pmatrix},
\quad
\phi=\frac14\begin{pmatrix} 1 & 0 \\ 0 & 1 \end{pmatrix},
\end{gather*}
and $P_1$ and~$P_2$ are first-order local homogeneous Hamiltonian operators
defined by the metrics through the Structure Formula (Corollary~\ref{cor:structure-g}) with
\begin{gather*}
Z_1=\begin{pmatrix} -\frac14u^2 & 1-\mu \\ \mu & 0 \end{pmatrix},
\quad
Z_2=\begin{pmatrix} 0 & \frac14u^1-\mu \\ \frac14u^1+\mu & \frac14u^2 \end{pmatrix},
\end{gather*}
respectively. Here, $\mu$ is an arbitrary parameter. Generally,
a matrix~$Z$ is defined up to such inessential constants in this construction.

\paragraph{Case $R^{(2)}$ from \cite{LSV:bi_hamil_kdv}: pencil
  $g_{\lambda,11}$.} The pair of Hamiltonian operators~$R+P_1$ and~$P_2$,
where $R$ is a third-order homogeneous operator defined by the Monge metric with
\begin{gather*}
\psi=\begin{pmatrix} 1 & 0 \\ u^2 & -u^1 \end{pmatrix},
\quad
\phi=\frac14\begin{pmatrix} 0 & 1 \\ 1 & 0 \end{pmatrix},
\end{gather*}
and $P_1$ and~$P_2$ are first-order homogeneous Hamiltonian operators defined
by the metrics through the Structure Formula with
\begin{gather*}
Z_1=\begin{pmatrix} -u^1 & u^2+\mu \\ -\mu u^1 & -\mu u^2 -1 \end{pmatrix},
\quad
Z_2=\begin{pmatrix} 0 & -u^1+\mu \\ -\mu u^1 & -\mu u^2 \end{pmatrix},
\end{gather*}
define the bi-Hamiltonian hierarchy
\[
u_{t_i}=(P_1+\varepsilon^2 R)\delta C_i,\quad i=1,2,
\]
where $C_1=\int_{S^1}u^1\mathrm dx$ and $C_2=\int_{S^1}\frac{u^2}{u^1}\mathrm dx$
are Casimirs of~$P_2$.

\paragraph{Case $R^{(2)}$ from \cite{LSV:bi_hamil_kdv}: pencil
  $g_{\lambda,13}$.} The pair of Hamiltonian operators~$R+P_1$ and~$P_2$,
where $R$ is a third-order homogeneous operator defined the Monge metric with
\begin{gather*}
\psi=\begin{pmatrix} 1 & 0 \\ u^2 & -u^1 \end{pmatrix},
\quad
\phi=\frac14\begin{pmatrix} 0 & 1 \\ 1 & 0 \end{pmatrix},
\end{gather*}
and $P_1$ and~$P_2$ are first-order homogeneous local Hamiltonian operators
defined by the metrics through the Structure Formula with
\begin{gather*}
Z_1=\begin{pmatrix} u^2 & u^1+\mu \\ -\mu u^1 & -\mu u^2 \end{pmatrix},
\quad
Z_2=\begin{pmatrix} 0 & \mu \\ -\mu u^1-1 & -\mu u^2 \end{pmatrix},
\end{gather*}
define the bi-Hamiltonian hierarchy
\[
u_{t_i}=(P_1+\varepsilon^2 R)\delta C_i,\quad i=1,2,
\]
where $C_1=\int_{S^1}\frac12(u^1)^2\mathrm dx$ and $C_2=\int_{S^1}u^2\mathrm dx$
are Casimirs of~$P_2$.

\paragraph{Case $R^{(3)}$ from \cite{LSV:bi_hamil_kdv}: pencil
  $g_{\lambda,12}$.} The pair of Hamiltonian operators~$R+P_1$ and~$P_2$,
where $R$ is a third-order homogeneous operator defined the Monge metric with
\begin{gather*}
\psi=\begin{pmatrix} 1 & 0 \\ u^2 & -u^1 \end{pmatrix},
\quad
\phi=\frac14\begin{pmatrix} 1 & 0 \\ 0 & 1 \end{pmatrix},
\end{gather*}
and $P_1$ and~$P_2$ are first-order homogeneous local Hamiltonian operators
defined by the metrics through the Structure Formula with metrics with
\begin{gather*}
Z_1=\begin{pmatrix} u^1-u^2 & -u^2+\mu \\ -\mu u^1+1 & 1-\mu u^2 \end{pmatrix},
\quad
Z_2=\begin{pmatrix} -1 & \mu \\ -\mu u^1-u^2 & -(\mu+2) u^2 \end{pmatrix},
\end{gather*}
define the bi-Hamiltonian hierarchy
\[
u_{t_i}=(P_1+\varepsilon^2 R)\delta C_i,\quad i=1,2,
\]
where $C_1=\int_{S^1}(u^1-u^2)\mathrm dx$ and
$C_2=\int_{S^1}\sqrt{(u^2)^2-2u^1u^2}\,\mathrm dx$ are Casimirs of~$P_2$.

\subsection{Bi-Hamiltonian structures of WDVV-type}

\paragraph{Hyperbolic Monge--Amp\`ere system}

The hyperbolic Monge--Amp\`ere equation $u_{tt}u_{xx}-u_{xt}^2=-1$ is a widely
studied differential equation. We adopt the viewpoint in
\cite{mokhov98:_sympl_poiss}, where the equation is rewritten as the
hydrodynamic-type system
\[u^1_t=u^2_x,\quad u^2_t=\left(\frac{(u^2)^2-1}{u^1}\right)_x\] via the change
of variables $u^1=u_{xx}$, $u^2=u_{xt}$ (see also~\cite{OpanasenkoVitolo2024}).
It admits a third-order homogeneous Hamiltonian operator $R$ determined by a
Monge metric whose factors are
\[
\psi=
\begin{pmatrix}
1 & 0 \\
-u^2 & u^1
\end{pmatrix},
\quad
\phi=
\begin{pmatrix}
0 & 1 \\
1 & 0
\end{pmatrix}.
\]
It also admits four compatible local first-order Hamiltonian operators that are defined by
\[
Z_1=
\begin{pmatrix}
0 & u^2 \\
0 & 0
\end{pmatrix},
\quad
Z_2=
\begin{pmatrix}
u^1 & 0 \\
0 & 1
\end{pmatrix},
\quad
Z_3=
\begin{pmatrix}
-u^2 & 0 \\
1 & 0
\end{pmatrix},
\quad
Z_4=
\begin{pmatrix}
0 & u^1 \\
0 & 0
\end{pmatrix}.
\]

\paragraph{A four-component example from
  \texorpdfstring{\cite{agafonov05:_integ}.}{Agafonov and Ferapontov.}}
It is believed that there is a unique integrable case within the class of
systems of conservation laws that admits a Hamiltonian formulation through a
third-order homogeneous Hamiltonian operator (see the discussion at the end
of~\cite{ferapontov18:_system_hamil}),
\begin{equation}\label{eq:44}
  \begin{split}
    &u^1_t=u^3_x,\\
&u^2_t=u^4_x,\\
&u^3_t=\left(\frac{u^1u^2u^4+u^3((u^3)^2+(u^4)^2-(u^2)^2-1)}{u^1u^3+u^2u^4}
\right)_x,\\
&u^4_t=\left(\frac{u^1u^2u^3+u^4((u^3)^2+(u^4)^2-(u^1)^2-1)}{u^1u^3+u^2u^4}
\right)_x.
\end{split}
\end{equation}
The above system was first derived in~\cite{agafonov05:_integ} by means of
projective-geometric considerations. In \cite{FPV17:_system_cl} a Lax pair and
a third-order homogeneous Hamiltonian operator~$R$ in canonical form were found
for the system. In~\cite{OpanasenkoVitolo2024} a compatible first-order
homogeneous Hamiltonian operator~$P$ was also found. The operator~$R$ is parameterised by a
Monge metric $f_{ij}=\sum_k\psi^k_i\psi^k_j$ ($\phi$ is the unit matrix),
\begin{equation*}
\psi=\begin{pmatrix}
-u^2 & u^1  & -u^4 & u^3 \\
-u^3 & -u^4 & u^1  & u^2 \\
1    & 0    & 0    & 0 \\
0    & 1    & 0    & 0
\end{pmatrix}.
\end{equation*}
The metric defining its first-order Hamiltonian operator of Ferapontov type is
parameterised by the two-by-two unit matrix~$c$ and $Z=-\frac12\psi$.
From now on the omit inessential parameters in~$Z$ in view of their number
to avoid unnecessarily cumbersome expressions.

\paragraph{WDVV in dimension \texorpdfstring{$N=3$}{N=3}.}
We recall that WDVV in arbitrary dimension, when rewritten as a first-order
quasilinear system of conservation laws, is endowed with a third-order
homogeneous Hamiltonian operator $R$ \cite{opanasenko25:_wdvv_hamil}.  In this
example and in the following WDVV examples we will construct what seems to be a
reasonable candidate for a first-order homogeneous operator of Ferapontov type
for WDVV systems in arbitrary dimension.

Let us start with $N=3$. The WDVV system that is associated with a Dubrovin
normal form $\eta=\eta^{(2)}$,
\begin{equation*}
  \eta^{(2)} =
  \begin{pmatrix}
    \mu & 0 & 1\\ 0 & 1 & 0\\ 1 & 0 & 0
  \end{pmatrix}
\end{equation*}
in dimension three takes the form
\begin{equation*}
\mu f_{ttt}f_{xxt}-f_{ttt}+f_{xxt}^2-f_{xxx}f_{xtt}-\mu f_{xxt}^2=0,
\end{equation*}
which in the new dependent variables $u^1=f_{xxx}$, $u^2=f_{xxt}$, $u^3=f_{xtt}$
can be rewritten as the hydrodynamic-type system
\begin{gather*}
(u^i)_t=(V^i)_x,\quad i=1,\dots3,
\end{gather*}
where $V^i=\psi^i_\alpha w_1^\alpha$,
\[
w_1=\begin{pmatrix}
0 & u^2 & u^3
\end{pmatrix},
\]
and
\[
\psi=(\psi^\alpha_i) = \begin{pmatrix}
-u^2 &\mu u^3+u^1  &-\mu u^2+1 \\
1                  &0  &0\\
0                  &1  &0\\
\end{pmatrix}.
\]
A third-order Hamiltonian operator is determined by the Monge metric
$f_{ij}=\psi_i^\alpha\phi_{\alpha\beta}\psi_j^\alpha$ with
\[
(\phi_{ij})=
\begin{pmatrix}
\mu & 1 & 0 \\
1 & 0 & 0 \\
0 & 0 & 1
\end{pmatrix},
\]
and the first-order operator~$P$ of Ferapontov type defined through the
Structure Formula (Corollary~\ref{cor:structure-g}) by
$Z^{\gamma j} = \eta^{\gamma j}_k u^k + \xi^{\gamma j}$,
\[
(Z^{\gamma j})=\begin{pmatrix}
 \frac32u^2&  u^3& 0& \\
 -\frac32-\mu u^2& u^1& 2u^2& \\
 -\mu u^3& u^2& 2u^3
\end{pmatrix},
\]
and the matrix~$c$,
\[
c=\mu
\begin{pmatrix}
1 & 0  \\
0 & -\mu
\end{pmatrix}.
\]
Note that $P$ becomes local when $\mu=0$, and it is equal to the operators
obtained in~\cite{FGMN97} and~\cite{vasicek21:_wdvv_hamil}.

\paragraph{WDVV in dimension
  \texorpdfstring{$N=4$ ($n=6$): new example.}{N=4 (n=6): new example.}}

The WDVV system that is associated with the Dubrovin normal form
$\eta=\eta^{(2)}$ in dimension four takes the form
\begin{equation*}
\begin{split}
&\mu (f_{yyz}f_{zzz}-f_{yzz}^2) + 2f_{yyz}f_{xyz} -f_{yyy}f_{xzz} -f_{xyy}f_{yzz}  = 0,\\
&\mu (f_{xzz}f_{yzz}-f_{zzz}f_{xyz})+f_{xxy}f_{yzz}  -f_{xxz}f_{yyz} + f_{zzz} + f_{xyy}f_{xzz} -f_{xyz}^2 = 0,\\
&\mu (f_{yyz}f_{xzz} - f_{xyz}f_{yzz})+f_{xxy}f_{yyz} -f_{xxz}f_{yyy} + f_{yzz} = 0,\\
&\mu (f_{xzz}^2-f_{xxz}f_{zzz})+f_{xxy}f_{xzz} -2f_{xxz}f_{xyz} + f_{xxx}f_{yzz} = 0,\\
&\mu (f_{xxz}f_{yzz} - f_{xzz}f_{xyz})+f_{xxz}f_{xyy} -f_{yyz}f_{xxx} + f_{xzz} = 0,\\
&\mu (f_{xxz}f_{yyz} - f_{xyz}^2)+f_{xxy}f_{xyy} -f_{xxx}f_{yyy} + 2f_{xyz} = 0
\end{split}
\end{equation*}
(only 4 equations here are functionally independent though,
see~\cite{opanasenko25:_wdvv_hamil}), which in the new dependent variables
$u^1=f_{xxx}$, $u^2=f_{xxy}$, $u^3=f_{xxz}$, $u^4=f_{xyy}$, $u^5=f_{xyz}$,
$u^6=f_{xzz}$ can be rewritten as the system of two
hydrodynamic-type systems
\begin{gather*}
(u^i)_y=(V^i)_x,\quad (u^i)_z=(W^i)_x \quad i=1,\dots6,
\end{gather*}
where $V^i=\psi^i_\alpha w_1^\alpha$, $W^i=\psi^i_\alpha w_2^\alpha$, $w$'s are determined to be
\[
w_1=\begin{pmatrix}
-2u^5 & -u^6 & u^2 & 0 & u^4 & u^5
\end{pmatrix},
\quad
w_2=\begin{pmatrix}
-u^6 & 0 & u^3 & 0 & u^5 & u^6
\end{pmatrix},
\]
and
\[
  \psi=(\psi^\alpha_i)=\begin{pmatrix}
    u^4 &0                    &-\mu u^5  &-u^1           &\mu u^3&0\\[1ex]
    u^5 &u^3     &-u^2-\mu u^6  &0        &-u^1&\mu u^3\\
    1                &0                    &0  &0         &0&0\\
    0                &-u^5                 &u^4  &-u^3        &\mu  u^6+u^2&-\mu  u^5+ 1\\
    0                &1                    &0    &0&0&0\\
    0 &0 &1 &0&0&0
\end{pmatrix}.
\]
Both these systems admit two compatible Hamiltonian operators of the
aforementioned form.  In particular, the two systems admits a third-order
homogeneous Hamiltonian operator $R$, which was found in
\cite{vasicek21:_wdvv_hamil}, defined by the Monge metric
$f_{ij}=\psi_i^\alpha\phi_{\alpha\beta}\psi_j^\alpha$ with
\[
(\phi_{\alpha\beta})=
\begin{pmatrix}
1 & 0 & 0 & 0 & 0 & 1 \\
0 & 0 & 0 & -\mu & -1 & 0 \\
0 & 0 & 0 & 1 & 0 & 0 \\
0 & -\mu & 1 & 0 & 0 & 0 \\
0 & -1 & 0 & 0 & 0 & 0 \\
1 & 0 & 0 & 0 & 0 & 2
\end{pmatrix}
\]
and the new first-order Hamiltonian operator of Ferapontov type (of the system $(u^i)_y=(V^i)_x$ only, for brevity) defined by
\[
(Z^{\gamma j}) =\begin{pmatrix}
 2u^1&  \mu u^6& 0& 0& -2 u^5& -2 u^6\\
 0& -u^3& 0& 0& -u^6& 0\\
 0& -\mu u^3& u^1& 0& u^2& 2u^3\\
 2u^3& u^5& u^6& 0& 0& 0\\
 0& -\mu u^5-1& u^2& 0& u^4& 2u^5\\
 0&-\mu u^6& u^3&0& u^5& 2u^6
\end{pmatrix}
\]
and the matrix~$c$,
\[
c=\mu
\begin{pmatrix}
0 & 1 & 0 \\
1 & 0 & 0 \\
0 & 0 & -\mu
\end{pmatrix}
\]
where the columns are associated with the unit matrix~$w_0$, and the vectors $w_1$ and~$w_2$, respectively.

\paragraph{WDVV in dimension \texorpdfstring{$N=5$ ($n=10$)}{
  N=5 (n=10)}: new examples.}
The WDVV system that is associated with a Dubrovin normal form in dimension in
the dependent variables (here we switch to the notation $(t^2,t^3,t^4,t^5)$ for
the independent variables):
\[u^1=f_{222},\ u^2=f_{223},\ u^3=f_{233},\ u^4=f_{224},\ u^5=f_{234},\ u^6=f_{244},\ u^7=f_{225},\ u^8=f_{235},\ u^9=f_{245},\ u^{10}=f_{255}\]
is rewritten as the system of three hydrodynamic-type systems; one of them has
the form
\begin{gather*}
(u^i)_3=(V^i)_2,\quad i=1,\dots10,
\end{gather*}
(it is unpractical to write the system in full details) where $V^i=\psi^i_\alpha w_2^\alpha$, $w$'s are
determined to be
\begin{gather*}
w_1=\begin{pmatrix}
u^1 & u^2 & 0 & u^4 & -u^7 & -u^8 & u^7 & 0 & 0& u^{10}
\end{pmatrix}\\
w_2=\begin{pmatrix}
u^2 & u^3 & -u^7 & u^5 & 0 & u^9 & u^8 & 0 & u^{10}& 0
\end{pmatrix}\\
w_3=\begin{pmatrix}
u^4 & u^5 & u^8 & u^6 & u^9 & 0 & u^9 & u^{10} & 0& 0
\end{pmatrix}\\
w_4=\begin{pmatrix}
u^7 & u^8 & 0 & u^9 & 0 & 0 & u^{10} & 0 & 0& 0
\end{pmatrix}
\end{gather*}
and
\[
\psi=(\psi^\alpha_i)=\begin{pmatrix}
1& 0& 0& 0& 0 & 0& 0& 0& 0& 0\\
0& 1& 0& 0& 0 & 0& 0& 0& 0& 0\\
-u^5& -u^3+u^4& u^2& -u^2& u^1 & 0 &\mu u^8& -\mu u^7& 0& 0\\
0& 0& 0& 1& 0 & 0& 0& 0& 0& 0\\
-u^6& -u^5& 0& 0& u^2 & u^1& \mu u^9-1& 0& -\mu u^7& 0\\
0& -u^6& -u^5& u^5& u^3-u^4 & u^2& 0& \mu u^9-1& -\mu u^8& 0\\
0& 0& 0& 0& 0 & 0& 1& 0& 0& 0\\
-u^9& -u^8& 0& -u^7& 0 & 0& \mu u^{10}+u^4& u^2& u^1& -\mu u^7\\
0& -u^9& -u^8& 0& -u^7 & 0& u^5& \mu u^{10}+u^3& u^2& -\mu u^8\\
0& 0& 0& -u^9& -u^8 & -u^7& u^6& u^5& \mu u^{10}+u^4& -\mu u^9+1
\end{pmatrix}.
\]
This system admits two compatible Hamiltonian operators of the
aforementioned form.  In particular, the third-order
Hamiltonian operator~$R$ (found in \cite{vasicek21:_wdvv_hamil}) is defined by the Monge
metric $f_{ij}=\psi_i^\alpha\phi_{\alpha\beta}\psi_j^\alpha$ with
\[
(\phi_{\alpha\beta})=
\begin{pmatrix}
0& 0& 0& 0& 0 & 0& 0& 0& 0& -1\\
0& 0& 0& 0& 0 & 0& 0& 0& -1& 0\\
0& 0& 0& 0& 0 & -1& 0& 0& 0& 0\\
0& 0& 0& 0& 0 & 0& 0& -1& 0& 0\\
0& 0& 0& 0& -1 & 0& 0& 0& 0& 0\\
0& 0& -1& 0& 0 & 0& 0& 0& 0& 0\\
0& 0& 0& 0& 0 & 0& -1& 0& 0& 0\\
0& 0& 0& -1& 0 & 0& 0& 0& 0& -\mu\\
0& -1& 0& 0& 0 & 0& 0& 0& -\mu& 0\\
-1& 0& 0& 0& 0 & 0& 0& -\mu& 0& 0
\end{pmatrix}
\]
and a first-order Hamiltonian operator of Ferapontov type, which is new and is
defined by the matrix
\[
(Z^{\gamma j})=\begin{pmatrix}
0& 0& -\mu u^7& -\mu u^7& 0 & 0& u^1& u^2& u^4& 2u^7\\
0& 0& -\mu u^8& -\mu u^8& 0 & 0& u^2& u^3& u^5& 2u^8\\
0& -u^1& -\frac32u^2& 0& 0 & 0& 0& -u^7& u^8& 0\\
0& 0& -\mu u^9& -\mu u^9-1& 0 & 0& u^4& u^5& u^6& 2u^9\\
-2u^1& -u^2& 0& 0& 0 & 0& -u^7& 0& u^9& 0\\
-2u^2& -u^3+u^4& \frac32u^5& 0& 0 & 0& -u^8& u^9& 0& 0\\
0& 0& -\mu u^{10}& -\mu u^{10}& 0 & 0& u^7& u^8& u^9& 2u^{10}\\
0& 0& 0& u^7& 0 & 0& 0& 0& u^{10}& 0\\
0& u^7& \frac32u^8& u^8& 0 & 0& 0& u^{10}& 0& 0\\
2u^7& u^8& 0& u^9& 0 & 0& u^{10}& 0& 0& 0
\end{pmatrix},
\]
and the matrix~$c$,
\[
c=\mu
\begin{pmatrix}
0 & 0& 1 & 0  \\
0 & 1& 0 & 0  \\
1 & 0& 0 & 0  \\
0 & 0& 0 & -\mu
\end{pmatrix}.
\]

\section*{Discussion}

So far, we have been able to derive a set of conditions whose vanishing is
equivalent to the vanishing of the Schouten bracket $[P,R_3]$. The strength of
our results lies in the fact that the Structure Formula in
Corollary~\ref{cor:structure-g} allows to reduce the system $[P,R_3]=0$ to
algebraic equations. The equations that determine a third-order homogeneous
operator in Doyle--Pot\"emin canonical form (see
Subsection~\ref{sec:third-order-hamilt}) are indeed algebraic, and by the
Structure Formula, a compatible first-order homogeneous operator is determined
by a finite number of constants.

If we consider a third-order operator as given, the equations are partly linear
(the equations~\eqref{eq:16} that determine Hamiltonian systems), and partly
algebraic of second degree (the symmetry of the velocity matrices of the
systems with respect to the metric~$g$ and the commutativity relations).

In our computational experiments, the additional conditions~\eqref{eq:221} and
\eqref{eq:118} seem not to play any role in the search for first-order
operators $P$ that are compatible with a given $R_3$, after having used all
other conditions. There exists the possibility that they vanish identically; we
will address this problem in future research work.

We are currently working on compatibility between second-order and first-order
homogeneous Hamiltonian operators \cite{lorenzoni:_compat12}; the calculations
that we did until now allow us to confirm that many of the results of this
paper hold in a strikingly similar way also in that case.


\providecommand{\cprime}{\/{\mathsurround=0pt$'$}}
  \providecommand*{\SortNoop}[1]{}

\section{Appendix}

Here we give the details of the proof of some results that we use to simplify
the coefficients of $[P,R]$.

\begin{lemma}\label{c_der_w_der}
$
S^{jq}_{\beta lp}+S^{jq}_{\beta pl}=0.
$
\end{lemma}

\begin{proof}
In order to get rid of derivatives with the help of~\eqref{c_der} and~\eqref{w_derivative},
we multiply $S^{jq}_{\beta lp}+S^{jq}_{\beta pl}$ by~$f^{kl}f^{ps}$,
\begin{gather*}
(w^j_{\beta k,h}c^{hq}_p -c^{jq}_{h,k}w^h_{\beta p} +c^{jq}_{k,h}w^h_{\beta p} -w^q_{\beta h,k}c^{hj}_p)f^{kl}f^{ps}+\\
(w^j_{\beta p,h}c^{hq}_k -c^{jq}_{h,p}w^h_{\beta k} +c^{jq}_{p,h}w^h_{\beta k} -w^q_{\beta h,p}c^{hj}_k)f^{kl}f^{ps}=\\
w^j_{\beta h,k}f^{kl}c^{hq}_p f^{ps}-c^{jq}_{h,k}f^{hp}w^s_{\beta p} f^{kl} +c^{jq}_{k,h}f^{kl}w^h_{\beta p} f^{ps} -w^q_{\beta h,k}f^{kl}c^{hj}_p f^{ps}+\\
w^j_{\beta h,p}f^{ps}c^{hq}_k f^{kl}-c^{jq}_{h,p}f^{hk}w^l_{\beta k} f^{ps} +c^{jq}_{p,h}f^{ps}w^h_{\beta k} f^{kl} -w^q_{\beta h,p}f^{ps}c^{hj}_k f^{kl}=\\
(-c^{lj}_kw^k_{\beta h}+w^l_{\beta k}c^{kj}_h)c^{hq}_p f^{ps}-(c^{jp}_hc^{hq}_k-c^{pj}_hc^{hq}_k-c^{pq}_h(c^{hj}_k+c^{jh}_k))w^s_{\beta p} f^{kl} \\
+(c^{jl}_kc^{kq}_h-c^{lj}_kc^{kq}_h-c^{lq}_k(c^{kj}_h+c^{jk}_h))w^h_{\beta p} f^{ps}-(-c^{lq}_kw^k_{\beta h}+w^l_{\beta k}c^{kq}_h)c^{hj}_p f^{ps}+\\
(-c^{sj}_pw^p_{\beta h}+w^s_{\beta p}c^{pj}_h)c^{hq}_k f^{kl}-(c^{jk}_hc^{hq}_p-c^{kj}_hc^{hq}_p-c^{kq}_h(c^{hj}_p+c^{jh}_p))w^l_{\beta k} f^{ps} \\
+(c^{js}_pc^{pq}_h-c^{sj}_pc^{pq}_h-c^{sq}_p(c^{pj}_h+c^{jp}_h))w^h_{\beta k} f^{kl}-(-c^{sq}_pw^p_{\beta h}+w^s_{\beta p}c^{pq}_h)c^{hj}_k f^{kl}
\end{gather*}
Now, we collect the coefficients of~$w^k_{\beta h}$, $w^s_{\beta h}$ and~$w^l_{\beta h}$ ($k$ is a dummy index, and~$s$ and~$l$ are fixed),
\begin{gather*}
\left(c^{hq}_p(c^{sj}_k f^{pl}+c^{lj}_kf^{sp})-c^{hj}_p(c^{lq}_kf^{sp}+c^{sq}_k f^{pl})\right)w^k_{\beta h}\\
+\left((c^{jl}_pc^{pq}_k-c^{lj}_pc^{pq}_k-c^{lq}_pc^{pj}_k-c^{lq}_pc^{jp}_k)f^{hk}
-(2c^{lq}_kc^{hj}_p-c^{jh}_pc^{lq}_k)f^{kp}-c^{hq}_pc^{lp}_kf^{kj}\right)w^s_{\beta h}\\
+\left((c^{js}_pc^{pq}_k-c^{sj}_pc^{pq}_k-c^{sq}_pc^{pj}_k-c^{sq}_pc^{jp}_k)f^{hk}
-(2c^{hj}_pc^{sq}_k-c^{jh}_pc^{sq}_k)f^{kp}-c^{hq}_pc^{sp}_kf^{kj}\right)w^l_{\beta h}
\end{gather*}
All three coefficients vanish. Indeed, in the coefficient of~$w^k_{\beta h}$ we use~\eqref{c_f_sym} and
renaming the dummy indices to get
\begin{gather*}
c^{hq}_p(c^{sj}_k f^{pl}+c^{lj}_kf^{sp})-c^{hj}_p(c^{lq}_kf^{sp}+c^{sq}_k f^{pl})=\\
c^{hq}_pc^{sj}_k f^{pl}+c^{hq}_pc^{lj}_kf^{sp}-c^{hj}_pc^{lq}_kf^{sp}-c^{hj}_pc^{sq}_k f^{pl}=\\
-\left(c^{lq}_pc^{sj}_k+c^{sq}_pc^{lj}_k-c^{sj}_pc^{lq}_k-c^{lj}_pc^{sq}_k\right)f^{ph}=\\
(-c^{lq}_kc^{sj}_p-c^{sq}_kc^{lj}_p+c^{sj}_pc^{lq}_k+c^{lj}_pc^{sq}_k)f^{ph}=0
\end{gather*}
The coefficients of~$w^s_{\beta h}$ and~$w^l_{\beta h}$ are the same up to permutation,
so we consider only one of them. As a preparatory step, we use again~\eqref{c_f_sym} and redenoting the dummy indices to yield,
\begin{gather*}
(c^{js}_pc^{pq}_k-c^{sj}_pc^{pq}_k-c^{sq}_pc^{pj}_k-c^{sq}_pc^{jp}_k)f^{hk}
-(2c^{hj}_pc^{sq}_k-c^{jh}_pc^{sq}_k)f^{kp}-c^{hq}_pc^{sp}_kf^{kj}=\\
(c^{js}_pc^{pq}_k-c^{sj}_pc^{pq}_k)f^{hk}
-c^{sq}_k(c^{jk}_pf^{ph} +c^{hj}_pf^{pk}+c^{kh}_pf^{pj})
-c^{hq}_pc^{sp}_kf^{kj},
\end{gather*}
where the middle summand vanishes due to~\eqref{c_f_cycle}, thus leaving us with
\begin{gather*}
(c^{js}_pc^{pq}_k-c^{sj}_pc^{pq}_k)f^{hk}-c^{hq}_pc^{sp}_kf^{kj}=\\
-c^{hq}_k(c^{js}_pf^{pk}+c^{kj}_pf^{ps}+c^{sk}_pf^{pj})=0
\end{gather*}
due to~\eqref{c_f_cycle} again.
\end{proof}

\begin{lemma}\label{lem:RR}
$R^{ijp}_{lm}+R^{ijp}_{ml}=0$,
where
\[
R^{ijp}_{lm}:=-\Gamma^{ij}_{k,l}c^{kp}_m
+\Gamma^{pj}_{k,l}c^{ki}_m
-\Gamma^{kj}_lc^{pi}_{k,m}+\Gamma^{kj}_lc^{pi}_{m,k}
+c^{\alpha\beta}(c^{ki}_lw^p_{\alpha k}w^j_{\beta m}-c^{kp}_lw^i_{\alpha k}w^j_{\beta m}).\]
\end{lemma}

\begin{proof}
First of all, we rewrite~\eqref{Gamma_derivative} as
\begin{gather*}
\Gamma^{ij}_{k,l} + c^{\alpha\beta}w^i_{\alpha k}w^j_{\beta l} = (\Gamma^{tj}_rc^{ri}_l - \Gamma^{rj}_l c^{ti}_r) f_{tk},
\end{gather*}
and simplify $R^{ijp}_{lm}+R^{ijp}_{ml}$ to
\begin{gather*}
-(\Gamma^{tj}_rc^{ri}_l - \Gamma^{rj}_l c^{ti}_r) f_{tk}c^{kp}_m
+(\Gamma^{tj}_rc^{rp}_m - \Gamma^{rj}_m c^{tp}_r) f_{tk}c^{ki}_l
-(\Gamma^{tj}_rc^{ri}_m - \Gamma^{rj}_m c^{ti}_r) f_{tk}c^{kp}_l
+(\Gamma^{tj}_rc^{rp}_l - \Gamma^{rj}_l c^{tp}_r) f_{tk}c^{ki}_m\\
-\Gamma^{kj}_lc^{pi}_{k,m}+\Gamma^{kj}_mc^{pi}_{l,k}-\Gamma^{kj}_mc^{pi}_{k,l}+\Gamma^{kj}_lc^{pi}_{m,k}
\end{gather*}
In order to get rid of the remaining derivatives with the help of~\eqref{c_der},
we multiply the above expression by~$f^{lx}f^{my}$,
\begin{gather*}
\Big((\Gamma^{tj}_rc^{ri}_l - \Gamma^{rj}_l c^{ti}_r) f_{mk}c^{kp}_t
-(\Gamma^{tj}_rc^{rp}_m - \Gamma^{rj}_m c^{tp}_r) f_{lk}c^{ki}_t
+(\Gamma^{tj}_rc^{ri}_m - \Gamma^{rj}_m c^{ti}_r) f_{lk}c^{kp}_t
-(\Gamma^{tj}_rc^{rp}_l - \Gamma^{rj}_l c^{tp}_r) f_{mk}c^{ki}_t\\
-\Gamma^{kj}_lc^{pi}_{k,m}+\Gamma^{kj}_mc^{pi}_{l,k}-\Gamma^{kj}_mc^{pi}_{k,l}+\Gamma^{kj}_lc^{pi}_{m,k}\Big)f^{lx}f^{my}=\\
(\Gamma^{tj}_rc^{ri}_l - \Gamma^{rj}_l c^{ti}_r) f_{mk}c^{kp}_tf^{lx}f^{my}
-(\Gamma^{tj}_rc^{rp}_m - \Gamma^{rj}_m c^{tp}_r) f_{lk}c^{ki}_tf^{lx}f^{my}\\
+(\Gamma^{tj}_rc^{ri}_m - \Gamma^{rj}_m c^{ti}_r) f_{lk}c^{kp}_tf^{lx}f^{my}
-(\Gamma^{tj}_rc^{rp}_l - \Gamma^{rj}_l c^{tp}_r) f_{mk}c^{ki}_tf^{lx}f^{my}\\
-(c^{pl}_kc^{ki}_m-c^{lp}_kc^{ki}_m-c^{li}_kc^{kp}_m-c^{li}_kc^{pk}_m)\Gamma^{xj}_lf^{my}
+(c^{px}_lc^{li}_k-c^{xp}_lc^{li}_k-c^{xi}_lc^{lp}_k-c^{xi}_lc^{pl}_k)\Gamma^{kj}_mf^{my}\\
-(c^{pm}_kc^{ki}_l-c^{mp}_kc^{ki}_l-c^{mi}_kc^{kp}_l-c^{mi}_kc^{pk}_l)\Gamma^{yj}_mf^{lx}
+(c^{py}_mc^{mi}_k-c^{yp}_mc^{mi}_k-c^{yi}_mc^{mp}_k-c^{yi}_mc^{pm}_k)\Gamma^{kj}_lf^{lx},
\end{gather*}
where we use~\eqref{c_f_sym} and~\eqref{Commut_Gamma_f} to reduce the expression to the form
with collectable coefficients of~$\Gamma^{xj}_l$ and~$\Gamma^{yj}_m$,
\begin{gather*}
(- c^{li}_k c^{kp}_mf^{my}-c^{pl}_kc^{ki}_mf^{my}-c^{yi}_mc^{pm}_kf^{lk}
+c^{lp}_kc^{ki}_mf^{my}+c^{li}_kc^{pk}_mf^{my}+c^{py}_mc^{mi}_kf^{lk})\Gamma^{xj}_l\\
+(- c^{mi}_k c^{kp}_lf^{lx}-c^{pm}_kc^{ki}_lf^{lx}+c^{mp}_kc^{ki}_lf^{lx}+c^{mi}_kc^{pk}_lf^{lx}
+c^{px}_mc^{li}_kf^{mk}
-c^{xi}_mc^{pl}_kf^{mk})\Gamma^{yj}_m
\end{gather*}
Both coefficients are the same up to permutation, and thus, we need to show that only one of them vanishes.
Again, we use~\eqref{c_f_sym} to reduce the expression to the form, where the cycle condition~\eqref{c_f_cycle} is evident,
\begin{gather*}
- c^{li}_k c^{kp}_mf^{my}-c^{pl}_kc^{ki}_mf^{my}-c^{yi}_kc^{pk}_mf^{lm}
-c^{lp}_mc^{yi}_kf^{mk}-c^{li}_kc^{yk}_mf^{mp}+c^{li}_kc^{ky}_mf^{mp}=\\
- c^{mi}_k c^{yp}_mf^{kl}-c^{yi}_kc^{kl}_mf^{mp}-c^{yi}_kc^{pk}_mf^{ml}
-c^{lp}_mc^{yi}_kf^{mk}-c^{li}_kc^{yk}_mf^{mp}+c^{li}_kc^{ky}_mf^{mp}
=\\
- c^{mi}_k c^{yp}_mf^{kl}+c^{yi}_kc^{lp}_mf^{mk}-c^{lp}_mc^{yi}_kf^{mk}-c^{li}_kc^{yk}_mf^{mp}+c^{li}_kc^{ky}_mf^{mp}=\\
c^{li}_k( c^{yp}_mf^{mk}+c^{pk}_mf^{my}+c^{ky}_mf^{mp})=0.
\end{gather*}
\end{proof}

\begin{lemma} The condition~\eqref{Gamma_c} can be rewritten in lower indices
  as
 \begin{equation}
   \label{eq:271}
   c_{mkl}\Gamma^{mj}_{\alpha i} + c_{mik}\Gamma^{mj}_{\alpha l}
   + c_{mli}\Gamma^{mj}_{\alpha k} = 0.
 \end{equation}
\end{lemma}

\begin{lemma}\label{lemma:re}
\begin{gather*}
-\Gamma^{ij}_{k,l}c^{pk}_m+\Gamma^{ij}_{k,m}c^{kp}_l
+\Gamma^{ij}_{k,m}c^{pk}_l-\Gamma^{pj}_{k,l}c^{ki}_m
-\Gamma^{ij}_{k,l}c^{kp}_m
+\Gamma^{pj}_{k,m}c^{ki}_l\\
+\Gamma^{ij}_kc^{kp}_{l,m}+\Gamma^{ij}_kc^{pk}_{l,m}-\Gamma^{ij}_kc^{pk}_{m,l}-\Gamma^{ij}_kc^{kp}_{m,l}-\Gamma^{kj}_lc^{pi}_{k,m}
+\Gamma^{pj}_kc^{ki}_{l,m}-\Gamma^{pj}_kc^{ki}_{m,l}+\Gamma^{kj}_mc^{pi}_{k,l}\\
-c^{ki}_mw^p_kw^j_l+c^{kp}_lw^i_kw^j_m+c^{pk}_lw^i_kw^j_m-c^{pk}_mw^i_kw^j_l+c^{ki}_lw^p_kw^j_m-c^{kp}_mw^i_kw^j_l=0
\end{gather*}
\end{lemma}

\begin{proof}
We start by removing all derivatives from the LHS by using identities~\eqref{c_der} and~\eqref{eq:8},
\begin{gather*}
-c^{pk}_mc^{ri}_l\Gamma^{sj}_kf_{sr}
+c^{kp}_lc^{ri}_m\Gamma^{sj}_kf_{sr}
+c^{pk}_lc^{ri}_m\Gamma^{sj}_kf_{sr}
-c^{kp}_mc^{ri}_l\Gamma^{sj}_kf_{sr}
-c^{ki}_mc^{rp}_l\Gamma^{sj}_kf_{sr}
+c^{ki}_lc^{rp}_m\Gamma^{sj}_kf_{sr}\\
-c^{pk}_mc^{ri}_k\Gamma^{sj}_lf_{sr}
+c^{kp}_lc^{ri}_k\Gamma^{sj}_mf_{sr}
+c^{pk}_lc^{ri}_k\Gamma^{sj}_mf_{sr}
-c^{kp}_mc^{ri}_k\Gamma^{sj}_lf_{sr}
-c^{ki}_mc^{rp}_k\Gamma^{sj}_lf_{sr}
+c^{ki}_lc^{rp}_k\Gamma^{sj}_mf_{sr}\\
+(c^{ps}_kc^{ri}_l+c^{si}_kc^{pr}_l+c^{si}_kc^{rp}_l)f_{rs}\Gamma^{kj}_m
-(c^{ps}_kc^{ri}_m+c^{si}_kc^{pr}_m+c^{si}_kc^{rp}_m)f_{rs}\Gamma^{kj}_l
  +(c^{si}_lc^{rk}_m-c^{ri}_mc^{sk}_l)f_{rs}\Gamma^{pj}_k
\end{gather*}
where we used the condition~\eqref{c_der} in the form
\begin{equation*}
c^{ij}_{p,l}=(c^{is}_pc^{rj}_l+c^{sj}_pc^{ir}_l+c^{sj}_pc^{ri}_l)f_{rs},
\end{equation*}
which is achieved by using the conditions~\eqref{c_f_sym} and~\eqref{c_f_cycle} in~\eqref{c_der}.
Now using the condition~\eqref{c_f_sym} and renaming dummy indices
we can rewrite $\Gamma^{sj}_l$ and~$\Gamma^{sj}_m$ in the second line as $\Gamma^{kj}_l$ and~$\Gamma^{kj}_m$
and sum up these terms with the third line to get
\[
c^{pk}_rc^{ri}_l\Gamma^{sj}_mf_{ks}
-c^{ri}_mc^{pk}_r\Gamma^{sj}_lf_{sk}
  +(c^{si}_lc^{rk}_m-c^{ri}_mc^{sk}_l)f_{rs}\Gamma^{pj}_k
\]
thus reducing the whole expression to
\begin{gather}\label{eq:121}
-c^{ri}_l(c^{pk}_m\Gamma^{sj}_kf_{sr}
-c^{pk}_r\Gamma^{sj}_kf_{sm}
-c^{sk}_mf_{rs}\Gamma^{pj}_k)
-c^{ri}_m(c^{pk}_r\Gamma^{sj}_kf_{sl}
+c^{pk}_l\Gamma^{sj}_kf_{rs}
-c^{sk}_lf_{rs}\Gamma^{pj}_k)
\end{gather}
Lowering indices in~\eqref{eq:121} by $f_{ai}f_{bj}f_{cp}$ we get
\begin{gather*}
  - (c_{dcm}\Gamma^{dj}_k
  - c_{dck}\Gamma^{dj}_m
  - c_{dkm}\Gamma^{dj}_c)c_{asl}f^{sk}
  + (c_{dcl}\Gamma^{dj}_k
  - c_{dck}\Gamma^{dj}_l
  - c_{asm}\Gamma^{dj}_c)c_{asm}f^{sk},
\end{gather*}
which vanishes in view of~\eqref{eq:271}.
\end{proof}
\end{document}